\RequirePackage{fix-cm}
\documentclass[smallextended]{svjour3}       % onecolumn (second format)
\smartqed  % flush right qed marks, e.g. at end of proof
\usepackage{graphicx}
\usepackage{eurosym}
\usepackage{amsmath}
\usepackage{amsfonts}
\usepackage{amssymb}
\usepackage{hyperref}
\usepackage{setspace}
\usepackage{mathpazo}
\usepackage{stmaryrd}
\usepackage{latexsym}
\journalname{Quantum Information Processing}
\begin{document}

\title{New entanglement-assisted MDS quantum codes from constacyclic codes
%\thanks{This research is supported by Y\i ld\i z Technical University Scientific
%Research Projects Coordination Department. Project Number: 2016-01-03-DOP01.}
}

\titlerunning{EAQMDS codes from constacyclic codes}        % if too long for running head

\author{Mehmet E. Koroglu}


\institute{Mehmet E. Koroglu \at
              Y\i ld\i z Technical University, Department of Mathematics, Faculty of Art and Sciences, 34220, Esenler, Istanbul-Turkey \\
              \email{mkoroglu@yildiz.edu.tr}}

\date{Received: date / Accepted: date}

\maketitle

\begin{abstract}
%Quantum error-correcting codes were introduced for security of quantum information. 
Construction of good quantum codes via classical codes
is an important task for quantum information and quantum computing.
In this work, by virtue of a decomposition of the defining set of constacyclic
codes we have constructed eight new classes of entanglement-assisted quantum maximum distance separable codes.
%Moreover, we have constructed two new classes of maximal-entanglement entanglement-assisted quantum codes.
\end{abstract}

\section{Introduction}

Quantum error-correcting (QEC for brevity) codes were introduced for
security of quantum information. Construction of good quantum codes via
classical codes is a crucial task for quantum information and quantum
computing (see Refs. \cite%
{Ashikhmin,Calderbank,Calderbank1,ChenH,Ketkar,Qian,Qian1,Steane,Xiaoyan}
for example). A $q$-ary quantum code $Q,$ denoted by parameters $\llbracket  %
n,k,d\rrbracket  _{q},$ is a $q^{k} $ dimensional subspace of the Hilbert
space $\mathbb{C}^{q^{n}}.$ A quantum code $\mathcal{C}$ with parameters $%
\llbracket  n,k,d\rrbracket  _{q}$ satisfy the quantum Singleton bound: $%
k\leq n-2d+2$ (see \cite{Ketkar}). If $k=n-2d+2,$ then $\mathcal{C}$ is
called a quantum maximum-distance-separable (MDS) code. In recent years,
many researchers have been working to find quantum MDS codes via
constacyclic codes (for instance, see \cite%
{Chen,Guardia,Kai1,Krishna,Xiaoyan,Zhang}).

Entanglement-assisted quantum error correcting (EAQEC for short) codes use
pre-existing entanglement between the sender and receiver to improve
information rate. For further details about EAQEC for example, see \cite%
{Brun,Fujiwara,Grassl,Hsieh,Hsieh1,Lai,Li,Lu,Wilde}.

Recently, many papers have been devoted for obtaining EAQEC codes derived
from classical error correcting codes. Some of these papers can be
summarized as follows. In \cite{Xiaoyan}, based on classical quaternary
constacyclic codes, some parameters for quantum codes were obtained. In \cite%
{ChenJ}, a decomposition of the defining set of negacyclic codes has been
proposed and by virtue of the proposed decomposition four classes of EAQEC
codes have been constructed. Fan et al., have constructed five classes of
entanglement-assisted quantum MDS (EAQMDS for short) codes based on
classical MDS codes by exploiting one or more pre-shared maximally entangled
states \cite{Fan}. Qian and Zhang have constructed some new classes of
maximum distance separable (MDS) linear complementary dual (LCD) codes with
respect to Hermitian inner product and as an application, they have
constructed new families of MDS maximal EAQEC codes in \cite{Qian2}. In \cite%
{Lu1}, Lu et al. constructed six classes of $q$-ary EAQMDS codes based on
classical negacyclic MDS codes. In \cite{Guenda}, Guenda et al. have shown
that the number of shared pairs required to construct an EAQEC code is
related to the hull of the classical codes. Using this fact, they gave
methods to construct EAQEC codes requiring desirable amounts of
entanglement. Further, they constructed maximal entanglement EAQEC codes
from LCD codes.

In this paper, based on a decomposition of the defining set of constacyclic
codes we have obtained eight new families of EAQMDS codes as follows:

\begin{enumerate}
\item $\llbracket n,n-\frac{6}{5}\left( q-7\right) -4\lambda -1,\frac{3}{5}%
\left( q-7\right) +2\lambda +4;5\rrbracket_{q},$ where $n=\frac{q^{2}+1}{10}%
, $ $1\leq \lambda \leq \frac{q+3}{10},$ $q$ is odd and $q\equiv 7\left( mod%
\text{ }10\right) .$

\item $\llbracket n,n-\frac{4}{5}\left( 2q+1\right) -4\lambda +7,\frac{2}{5}%
\left( 2q+1\right) +2\lambda +2;9\rrbracket_{q},$ where $n=\frac{q^{2}+1}{10}%
,$ $1\leq \lambda \leq \frac{q+3}{10},$ $q$ is odd and $q\equiv 7\left( mod%
\text{ }10\right) .$

\item $\llbracket n,n-\frac{6}{5}\left( q-3\right) -4\lambda +3,\frac{3}{5}%
\left( q-3\right) +2\lambda +2;5\rrbracket_{q},$ where $n=\frac{q^{2}+1}{10}%
, $ $1\leq \lambda \leq \frac{q-3}{10},$ $q$ is odd and $q\equiv 3\left( mod%
\text{ }10\right) .$

\item $\llbracket n,n-\frac{8}{5}\left( q-3\right) -4\lambda +7,\frac{4}{5}%
\left( q-3\right) +2\lambda +2;9\rrbracket_{q},$ where $n=\frac{q^{2}+1}{10}%
, $ $1\leq \lambda \leq \frac{q-3}{10},$ $q$ is odd and $q\equiv 3\left( mod%
\text{ }10\right) .$

\item $\llbracket n,n-\frac{6}{5}\left( q-2\right) -4\lambda +4,\frac{3}{5}%
\left( q-2\right) +2\lambda +1;4\rrbracket_{q},$ where $n=\frac{q^{2}+1}{5},$
$1\leq \lambda \leq \frac{q+3}{5},$ $q=2^{e}$ and $q\equiv 2\left( mod\text{
}10\right) .$

\item $\llbracket n,n-\frac{2}{5}\left( 3q-14\right) -4\lambda ,\frac{\left(
3q-14\right) }{5}+2\lambda +3;4\rrbracket_{q},$ where $n=\frac{q^{2}+1}{5},$
$1\leq \lambda \leq \frac{q+2}{5},$ $q=2^{e}$ and $q\equiv 8\left( mod\text{
}10\right) .$

\item $\llbracket n,n-\frac{6}{5}\left( q-2\right) -4\lambda +4,\frac{3}{5}%
\left( q-2\right) +2\lambda +1;4\rrbracket_{q},$ where $n=\frac{q^{2}+1}{13}%
, $ $1\leq \lambda \leq \frac{q+3}{5},$ $q=2^{e}$ and $q\equiv 5\left( mod%
\text{ }13\right) .$

\item $\llbracket n,n-\frac{6}{5}\left( q-4\right) -4\lambda -8,\frac{3}{5}%
\left( q-4\right) +2\lambda +4;4\rrbracket_{q},$ where $n=\frac{q^{2}+1}{17}%
, $ $1\leq \lambda \leq \frac{q+4}{17},$ $q=2^{e}$ and $q\equiv 13\left( mod%
\text{ }17\right) .$
\end{enumerate}

The rest of the paper is organized as follows. In Sect. 2, we review basics
about linear codes and constacyclic codes. In Sect. 3, we review some basics
about EAQEC codes. In Sect. 4 and Sect. 5, we define a decomposition of the
defining set of constacyclic codes and based on this method we construct
eight families of EAQMDS codes. The last section contains some comparative
results and concludes this paper.

\section{Basics about constacyclic codes}

In this section, we review some preliminaries of constacyclic codes. For
further and detailed information, readers may refer to \cite%
{Guardia,Krishna,Kai,Chen,Zhang}.

For given a positive integer $e$ and prime number $p,$ let $q=p^{e}$ and $%
\mathbb{F}_{q^{2}}$ be the finite field of $q^{2}$ elements. The Hermitian
inner product of $u=\left( u_{0},\ldots ,u_{n-1}\right) $ and $v=\left(
v_{0},\ldots ,v_{n-1}\right) \in \mathbb{F}_{q^{2}}^{n}$ is defined to be $%
\left\langle u,v\right\rangle _{H}=\sum\limits_{i=0}^{n-1}u_{i}v_{i}^{q}.$
If $\mathcal{C} $ is a $k$-dimensional subspace of $\mathbb{F}_{q^{2}}^{n},$
then $\mathcal{C}$ is called as a $q^{2}$-ary linear code of length $n$ and
dimension $k$ and denoted by $\left[ n,k\right] _{q^{2}}.$ The weight $%
wt\left( c\right) $ of a codeword $c\in \mathcal{C}$ is defined as the
number of its nonzero coordinates. The minimum nonzero weight $d$ amongst
all codewords of $\mathcal{C}$ is said to be the minimum weight of $\mathcal{%
C}.$ A linear code $\mathcal{C}$ of length $n $ is said to be constacyclic
if for any codeword $\left( c_{0},\ldots ,c_{n-1}\right) \in \mathcal{C}$ we
have that $\left( \alpha c_{n-1},\ldots ,c_{n-2}\right) \in \mathcal{C},$
where $0\neq \alpha \in \mathbb{F}_{q^{2}}.$ It can be seen that $xc\left(
x\right) $ corresponds to a constacyclic shift of $c\left( x\right) $ in the
quotient ring $\mathbb{F}_{q^{2}}\left[ x\right] /\left\langle x^{n}-\alpha
\right\rangle ,$ where $c\left( x\right) =c_{0}+c_{1}x+\ldots
+c_{n-1}x^{n-1}.$ Then, a $q^{2} $-ary constacyclic code $\mathcal{C}$ of
length $n$ is an ideal of $\mathbb{F}_{q^{2}}\left[ x\right] /\left\langle
x^{n}-\alpha \right\rangle $ and $\mathcal{C}$ is generated by a monic
polynomial $g\left( x\right) $ such that $g\left( x\right) |\left(
x^{n}-\alpha \right) .$ If $\gcd \left( q,n\right) =1,$ then $x^{n}-\alpha $
doesn't have multiple roots.

Let $m$ be the multiplicative order of $q^{2}$ in modulo $rn,$ where $r=q+1,$
and suppose that $\delta $ is a primitive $rn^{th}$ root of unity in $%
\mathbb{F}_{q^{2}}^{\ast }$ such that $\delta ^{n}=\alpha .$ Let $\zeta
=\delta ^{r}$, then $\zeta $ is a primitive $n^{th}$ root of unity.
Therefore, the roots of $x^{n}-\alpha $ are $\left\{ \delta ,\delta
^{1+r},\ldots ,\delta ^{1+r(n-1)}\right\} .$ Hence, it follows that $%
x^{n}-\alpha =\prod\limits_{i=0}^{n-1}\left( x-\zeta ^{ri}\right) .$

The $q^{2}$-cyclotomic coset of $i$ modulo $rn$ is defined by
\[
\mathcal{C}_{i}=\left\{ \left. iq^{2j}\left( {mod}\text{ }rn\right)
\right\vert j\in \mathbb{Z}\right\}.
\]

The Hermitian dual of a linear code $\mathcal{C}$ of length $n$ is defined
as $\mathcal{C}^{\bot _{H}}=\left\{ \left. u\in \mathbb{F}%
_{q^{2}}^{n}\right\vert \left\langle u,v\right\rangle _{H}=0\mathcal{\ }%
\text{for all }v\in \mathcal{C}\right\}.$ A $q^{2}$-ary linear code $%
\mathcal{C}$ of length $n$ is called Hermitian self-orthogonal if $\mathcal{C%
}\subseteq \mathcal{C}^{\bot _{H}}.$

Let $\mathcal{O}_{rn}=\left\{ \left. 1+rj\right\vert 0\leq j\leq n-1\right\}
.$ Then, the defining set of a constacyclic code $\mathcal{C=}\left\langle
g\left( x\right) \right\rangle $ of length $n$ is the set $Z=\left\{ \left.
i\in \mathcal{O}_{rn}\right\vert \delta ^{i}\text{ is a root of }g\left(
x\right) \right\} .$ If $\mathcal{C}$ is an $\left[ n,k\right] _{q^{2}}$ $%
\alpha $-constacyclic code with defining set $Z,$ then the Hermitian dual $%
\mathcal{C}^{\bot _{H}}$ of $\mathcal{C}$ is an $\alpha ^{-q}$-constacyclic
code with defining set $Z^{\bot _{H}}=\left\{ \left. z\in \mathcal{O}%
_{rn}\right\vert -qz\text{ }\left( {mod}\text{ }rn\right) \notin Z\right\} .$

As in cyclic codes, there exists BCH bound for $\alpha $-constacyclic (see
\cite{Aydin,Krishna}) as follows.

\begin{proposition}
\label{prop1} \cite{Aydin,Krishna} \emph{(The BCH bound for constacyclic
codes)} Let $\mathcal{C=}\left\langle g\left( x\right) \right\rangle $ be a $%
q^{2}$-ary $\alpha $-constacyclic code of length $n,$ where $\alpha $ is an
primitive $r^{th}$ root of unity. If the polynomial $g\left( x\right)$ has
the elements $\left\{ {\left. \delta ^{{1+}r{j}}\right\vert l\leq j\leq l+d-2%
}\right\}$ as the roots, where $\delta $ is a $rn^{th}$ primitive root of
unity with $\delta ^{n}=\alpha .$ Then, the minimum distance of $\mathcal{C}$
is at least $d$.
\end{proposition}

The following proposition give a criterion to determine whether or not an $%
\alpha $-constacyclic code of length $n$ over $\mathbb{F}_{q^{2}}$ is
Hermitian dual containing (see \cite{Kai1} Lemma 2.2).

\begin{proposition}
\label{prop2} Let $\alpha \in \mathbb{F}_{q^{2}}^{\ast }$ be of order $r. $
If $\mathcal{C}$ is an $\alpha $-constacyclic code of length $n$ over $%
\mathbb{F}_{q^{2}}$ with defining set $Z\subseteq \mathcal{O}_{rn},$ then $%
\mathcal{C} $ contains its Hermitian dual code if and only if $Z\cap \left(
-qZ\right) =\emptyset ,$ where $-qZ=\left\{ \left. -qz\left( {mod}\text{ }%
rn\right) \right\vert z\in Z\right\} .$
\end{proposition}

\section{Basics about entanglement-assisted quantum codes}

In this section, we review some basic notions and results of EAQEC codes.
The following result is about the Singleton bound of classical linear codes.

\begin{proposition}
\label{prop3} \cite{MacWilliams} \emph{(Singleton bound)} If an $\left[ n,k,d%
\right] $ linear code $\mathcal{C}$ over $\mathbb{F}_{q}$ exists, then $%
k\leq n-d+1.$ If $k=n-d+1,$ then $\mathcal{C}$ is called an MDS code.
\end{proposition}

Let $H$ be an $\left( n-k\right) \times n$ parity check matrix of $\mathcal{C%
}$ over $\mathbb{F}_{q^{2}}.$ Then, $\mathcal{C}^{\bot _{H}}$ has an $%
n\times \left( n-k\right) $ generator matrix $H^{\ast },$ where $H^{\ast }$
is the conjugate transpose matrix of $H$ over $\mathbb{F}_{q^{2}}.$

The following is called the Hermitian method and it enable us to construct
EAQEC codes from classical linear codes.

\begin{theorem}
\label{th1} \cite{Lu} If $\mathcal{C}$ is a classical code and $H$ is its
parity check matrix over $\mathbb{F}_{q^{2}},$ then there exists EAQEC codes
with parameters $\llbracket  n,2k-n+c,d;c \rrbracket  _{q},$ where $%
c=rank\left( HH^{\ast }\right) .$
\end{theorem}

\begin{proposition}
\label{prop4} \cite{Brun,Grassl} Assume that $\mathcal{C}$ is an EAQEC code
with parameters $\llbracket  n,k,d;c\rrbracket  _{q}$, if $d\leq (n+2)/2,$
then $\mathcal{C} $ satisfies the entanglement-assisted Singleton bound $%
n+c-k\geq 2(d-1).$ If $\mathcal{C}$ satisfies the equality $n+c-k=2(d-1)$
for $d\leq (n+2)/2,$ then it is called an EAQMDS code.
\end{proposition}

A definition for decomposition of the defining set of cyclic codes was given
in \cite{Lu}. In the following, we give a decomposition of the defining set
of constacyclic codes, which is the same as negacyclic case defined by Chen
et al. in \cite{ChenJ}.

\begin{definition}
Let $\alpha\in \mathbb{F}_{q^{2}}^{*}$ be a primitive $r^{th}$ root of unity
and $\mathcal{C}$ be an $\alpha$-constacyclic code of length $n$ with
defining set $Z.$ Assume that $Z_{1}=Z\cap (-qZ)$ and $Z_{2}=Z\backslash
Z_{1},$ where $-qZ=\left\{ {rn-qx|x\in Z}\right\}$ $r$ is a factor of $q+1.$
Then, $Z=Z_{1}\cup Z_{2}$ is called a decomposition of the defining set of $%
\mathcal{C}$.
\end{definition}

In \cite{ChenJ}, Chen et al. showed that the number of entangled states
required for negacyclic codes is $c=|Z_{1}|,$ which is the same for
constacyclic codes.

\begin{lemma}
\label{lm1} Let $\mathcal{C}$ be an $\alpha$-constacyclic code of length $n$
over $\mathbb{F}_{q^{2}},$ where $gcd(n,q)=1.$ Suppose that $Z$ is the
defining set of the $\alpha$-constacyclic code $\mathcal{C}$ and $%
Z=Z_{1}\cup Z_{2}$ is a decomposition of $Z.$ Then, the number of entangled
states required is $c=|Z_{1}|.$
\end{lemma}

\section{Construction of EAQMDS codes from constacyclic codes ($q$ is odd)}

Throughout this section, $q$ is an odd prime, $r=q+1$ and $s=\frac{q^{2}+1}{2%
}.$ The multiplicative order of $q$ modulo $n$ is denoted by $ord_{n}\left(
q\right) .$ Let $\alpha \in \mathbb{F}_{q^{2}}^{\ast }$ be a primitive $%
r^{th}$ root of unity. Here, we need to emphasize that the codes we give in
this section are different from the codes given in \cite{Lu2}, because they
obtained some parameters with the number of entangled states $c=1.$

\subsection{EAQMDS codes of length $n=\frac{q^{2}+1}{10},$ where $q\equiv
7\left( mod\text{ }10\right) $}

Note that $n=\frac{q^{2}+1}{10},$ and so $ord_{rn}\left( q^{2}\right) =2.$
This means that each $q^{2}$-cyclotomic coset modulo $rn$ includes one or
two elements. Let $q\equiv 7\left( mod\text{ }10\right) $ and $s=\frac{%
q^{2}+1}{2}.$ It can be easily seen that the $q^{2}$-cyclotomic cosets
modulo $rn$ containing some integers from $1$ to $rn$ are $C_{s}=\left\{
s\right\} $ and $C_{s-rj}=\left\{ s-rj,s+rj\right\} ,$ where $1\leq j\leq
\frac{q+1}{2}.$

\begin{lemma}
\label{lm5*} Let $q\equiv 7\left( mod\text{ }10\right) .$ If $\mathcal{C}$
is a $q^{2}$-ary constacyclic code of length $n$ and defining set $Z=\cup
_{j=1}^{\lambda }C_{s-rj},$ where $1\leq \lambda \leq \frac{3\left(
q-7\right) }{10}+1,$ then $\mathcal{C}^{\perp _{H}}\subseteq \mathcal{C}.$
\end{lemma}

\begin{proof}
From Proposition \ref{prop2}, it is sufficient to prove that $Z\cap \left(
-qZ\right) =\emptyset .$ Assume that $Z\cap \left( -qZ\right) \neq \emptyset
.$ Then, there exists two integers $j,k,$ where $1\leq j,k\leq \frac{3\left(
q-7\right) }{10}+1,$ such that $s-rj\equiv -q\left( s-rk\right) \left( mod%
\text{ }rn\right) $ or $s-rj\equiv -q\left( s+rk\right) $ $\left( mod\text{ }%
rn\right) .$

\textbf{Case 1:} Let $s-rj\equiv -q\left( s-rk\right) \left( mod\text{ }%
rn\right) .$ This is equivalent to $s\equiv j+qk\left( mod\text{ }n\right) .$
As $s\equiv 0\left( mod\text{ }n\right) ,$ we get $j+qk\equiv 0\left( mod%
\text{ }n\right) .$ Since $1\leq j,k\leq \frac{3\left( q-7\right) }{10}+1,$ $%
q+1\leq j+qk\leq \frac{3\left( q-7\right) }{10}+1+q\left( \frac{3\left(
q-7\right) }{10}+1\right) <3n.$ Then, we have that $j+qk\equiv n\left( mod%
\text{ }n\right) $ or $j+qk\equiv 2n\left( mod\text{ }n\right) .$ If $%
j+qk=n, $ then $j+qk=\frac{q^{2}+1}{10}=q\frac{\left( q-7\right) }{10}+\frac{%
7q+1}{10}.$ By division algorithm, $j=\frac{7q+1}{10}.$ This is a
contradiction, because $0\leq j\leq \frac{3\left( q-7\right) }{10}.$ If $%
j+qk=2n,$ then $j+qk=\frac{2\left( q^{2}+1\right) }{10}=q\frac{2\left(
q-7\right) }{10}+\frac{2\left( 7q+1\right) }{10}.$ From division algorithm, $%
j=\frac{2\left( 7q+1\right) }{10}.$ This contradicts with the fact $0\leq
j\leq \frac{3\left( q-7\right) }{10}.$

\textbf{Case 2:} Let $s-rj\equiv -q\left( s+rk\right) \left( mod\text{ }%
rn\right) .$ This is equivalent to $s\equiv j-qk\left( mod\text{ }n\right) .$
Since $s\equiv 0\left( mod\text{ }n\right) ,$ we have $j-qk\equiv 0\left( mod%
\text{ }n\right) .$ Also we have $1\leq j,k\leq \frac{3\left( q-7\right) }{10%
}.$ This results in $-3n<1-3q\left( \frac{3\left( q-7\right) }{10}+1\right)
\leq j-qk\leq 1-q<0.$ Thus, the solution is $j-qk\equiv -2n\left( mod\text{ }%
n\right) $ or $j-qk\equiv -n\left( mod\text{ }n\right) .$ If $j-qk=-2n,$
then $j-qk=\frac{-2\left( q^{2}+1\right) }{10}=-q\frac{2\left( q-7\right) }{%
10}+\frac{2\left( 7q-1\right) }{10}.$ From division algorithm, $j=\frac{%
2\left( 7q-1\right) }{10}.$ This is a contradiction, because $1\leq j\leq
\frac{3\left( q-7\right) }{10}+1.$ If $j-qk=-n,$ then $j-qk=\frac{-\left(
q^{2}+1\right) }{10}=-q\frac{\left( q-7\right) }{10}-\frac{\left(
7q+1\right) }{10}.$ By division algorithm, $j=\frac{-\left( 7q+1\right) }{10}%
\equiv \frac{q\left( q-7\right) }{10}\left( mod\text{ }n\right) .$ This is a
contradiction, because $0\leq j\leq \frac{3\left( q-7\right) }{10}.$
\end{proof}

\begin{lemma}
\label{lm6*} Let $q\equiv 7\left( mod\text{ }10\right) .$ If $\mathcal{C}$
is a $q^{2}$-ary constacyclic code of length $n$ and defining set $\overline{%
Z}=\cup _{j=\frac{3\left( q-7\right) }{10}}^{t}C_{s-rj},$ where $1\leq t\leq
\frac{\left( q-7\right) }{10},$ then $\mathcal{C}^{\perp _{H}}\subseteq
\mathcal{C}.$
\end{lemma}

\begin{proof}
The proof for this lemma is very similar to the proof of Lemma \ref{lm5*}.
\end{proof}

As an immediate result of Lemma \ref{lm6*} we have $\overline{Z}\cap \left(
-q\overline{Z}\right) =\emptyset .$

\begin{lemma}
\label{lm7*} Let $q\equiv 7\left( mod\text{ }10\right) $ and $s=\frac{q^{2}+1%
}{2}.$ Then we have the following:

\begin{enumerate}
\item $-qC_{s}=C_{s}=\left\{ s\right\} ,$

\item $-qC_{s-r\frac{\left( q+3\right) }{10}}=C_{s-r\left( \frac{3q-1}{10}%
\right) }=\left\{ s-r\left( \frac{3q-1}{10}\right) ,s+r\left( \frac{3q-1}{10}%
\right) \right\} ,$

\item $-qC_{s-r\left( \frac{2q-4}{10}\right) }=C_{s-r\left( \frac{2q+1}{5}%
\right) }=\left\{ s-r\frac{\left( 2q+1\right) }{5},s+r\frac{\left(
2q+1\right) }{5}\right\} .$
\end{enumerate}
\end{lemma}

\begin{proof}
\begin{enumerate}
\item Since $s=\frac{q^{2}+1}{2}=5n,$ we have $%
-qs=-5qn=-5qn-5n+5n=-(q+1)5n+5n\equiv 5n\left( mod\text{ }rn\right) .$ This
implies that $-qC_{s}=C_{s}.$

\item Observe that $-q\left( s-r\frac{\left( q+3\right) }{10}\right) =-qs+qr%
\frac{\left( q+3\right) }{10}\equiv 5n+qr\frac{\left( q+3\right) }{10}\left(
mod\text{ }rn\right) .$ We conclude that $5n+r\left( \frac{q^{2}+1}{10}+%
\frac{3q-1}{10}\right) \equiv s+r\frac{\left( 3q-1\right) }{10}\left( mod%
\text{ }rn\right) .$

\item It is enough to show that $-q\left( s-r\left( \frac{2q-4}{10}\right)
\right) \equiv s-r\frac{\left( 2q+1\right) }{5}\left( mod\text{ }rn\right) .$
It follows that $-q\left( s-r\left( \frac{2q-4}{10}\right) \right) \equiv
-qs+rq\left( \frac{2q-4}{10}\right) \equiv -qs+r\left( \frac{2\left(
q^{2}+1\right) }{10}-\frac{4q+2}{10}\right) $ $\left( mod\text{ }rn\right) ,$
and so we have $-q\left( s-r\left( \frac{2q-4}{10}\right) \right) \equiv s-r%
\frac{\left( 2q+1\right) }{5}\left( mod\text{ }rn\right) .$
\end{enumerate}
\end{proof}

In Theorem \ref{th4*}, we give a class of EAQMDS codes of length $n=\frac{%
q^{2}+1}{10}$ and with entangled states $c=5.$

\begin{theorem}
\label{th4*} Let $q\equiv 3\left( mod\text{ }10\right) .$ If $\mathcal{C}$
is an $q^{2}$-ary $\alpha $-constacyclic code of length $n$ with defining
set $Z=\cup _{j=0}^{\frac{3\left( q-7\right) }{10}+1+\lambda }C_{s-rj},$
then there exists EAQMDS codes with parameters
\[
\llbracket n,n-\frac{6}{5}\left( q-7\right) -4\lambda -1,\frac{3}{5}\left(
q-7\right) +2\lambda +4;5\rrbracket_{q},
\]%
where $1\leq \lambda \leq \frac{\left( q+3\right) }{10}.$
\end{theorem}

\begin{proof}
Since the defining set of given $\alpha $-constacyclic code $\mathcal{C}$ of
length $n$ is $Z=\cup _{j=0}^{\frac{3}{10}\left( q-7\right) +1+\lambda
}C_{s-rj},$ and the cardinality $\left\vert Z\right\vert =\frac{3}{5}\left(
q-7\right) +2\lambda +3,$ then from Proposition \ref{prop1} and Proposition %
\ref{prop3}, $\mathcal{C}$ is a $q^{2}$-ary MDS $\alpha $-constacyclic code
with parameters
\[
\left[ n,n-\left( \frac{3}{5}\left( q-7\right) +2\lambda +3\right) ,\frac{3}{%
5}\left( q-7\right) +2\lambda +4\right] _{q^{2}}.
\]%
Thus, we have the following:%
\begin{eqnarray*}
Z_{1} &=&Z\cap (-qZ)= \\
&&\left( \left( \cup _{j=0}^{\frac{3}{10}\left( q-7\right)
+1}C_{s-rj}\right) \cup \left( \cup _{j=\frac{3}{10}\left( q-7\right) +2}^{%
\frac{3}{10}\left( q-7\right) +1+\lambda }C_{s-rj}\right) \right) \cap \\
&&\left( -q\left( \cup _{j=0}^{\frac{3}{10}\left( q-7\right)
+1}C_{s-rj}\right) \cup -q\left( \cup _{j=\frac{3}{10}\left( q-7\right) +7}^{%
\frac{3}{10}\left( q-7\right) +1+\lambda }C_{s-rj}\right) \right) \\
&=&\left( \left( \cup _{j=0}^{\frac{3}{10}\left( q-7\right)
+1}C_{s-rj}\right) \cap -q\left( \cup _{j=0}^{\frac{3}{10}\left( q-7\right)
+1}C_{s-rj}\right) \right) \cup \\
&&\left( \left( \cup _{j=0}^{\frac{3}{10}\left( q-7\right)
+1}C_{s-rj}\right) \cap -q\left( \cup _{j=\frac{3}{10}\left( q-7\right) +2}^{%
\frac{3}{10}\left( q-7\right) +1+\lambda }C_{s-rj}\right) \right) \\
&&\cup \left( \left( \cup _{j=\frac{3}{10}\left( q-7\right) +2}^{\frac{3}{10}%
\left( q-7\right) +1+\lambda }C_{s-rj}\right) \cap -q\left( \cup _{j=0}^{%
\frac{3}{10}\left( q-7\right) +1}C_{s-rj}\right) \right) \\
&&\cup \left( \left( \cup _{j=\frac{3}{10}\left( q-7\right) +2}^{\frac{3}{10}%
\left( q-7\right) +1+\lambda }C_{s-rj}\right) \cap -q\left( \cup _{j=\frac{3%
}{10}\left( q-7\right) +2}^{\frac{3}{10}\left( q-7\right) +1+\lambda
}C_{s-rj}\right) \right) .
\end{eqnarray*}%
We claim that
\[
Z_{1}=Z\cap (-qZ)=C_{s}\cup C_{s-r\frac{3q-1}{10}}\cup C_{s-r\frac{\left(
q+3\right) }{10}}.
\]%
From Lemma \ref{lm5*}, we have $\left( \cup _{j=1}^{\frac{3}{10}\left(
q-7\right) +1}C_{s-rj}\right) \cap -q\left( \cup _{j=1}^{\frac{3}{10}\left(
q-7\right) +1}C_{s-rj}\right) =\emptyset .$ By examining the coset structure
of the defining set $Z,$ we can see that if $j=0,$ then $\ C_{s}\cap
-qC_{s}=\left\{ s\right\} .$

Hence, we need to show that%
\[
\left( \cup _{j=1}^{\frac{3}{10}\left( q-7\right) +1}C_{s-rj}\right) \cap
-q\left( \cup _{j=\frac{3}{10}\left( q-7\right) +2}^{\frac{3}{10}\left(
q-7\right) +1+\lambda }C_{s-rj}\right) =C_{s-r\frac{\left( q+3\right) }{10}%
},
\]%
\[
\left( \cup _{j=\frac{3}{10}\left( q-7\right) +2}^{\frac{3}{10}\left(
q-7\right) +1+\lambda }C_{s-rj}\right) \cap -q\left( \cup _{j=1}^{\frac{3}{10%
}\left( q-7\right) +1}C_{s-rj}\right) =C_{s-r\frac{3q-1}{10}},
\]%
\[
\left( \cup _{j=\frac{3}{10}\left( q-7\right) +2}^{\frac{3}{10}\left(
q-7\right) +1+\lambda }C_{s-rj}\right) \cap -q\left( \cup _{j=\frac{3}{10}%
\left( q-7\right) +2}^{\frac{3}{10}\left( q-7\right) +1+\lambda
}C_{s-rj}\right) =\emptyset .
\]

We first show that
\[
\left( \cup _{j=\frac{3}{10}\left( q-7\right) +2}^{\frac{3}{10}\left(
q-7\right) +1+\lambda }C_{s-rj}\right) \cap -q\left( \cup _{j=1}^{\frac{3}{10%
}\left( q-7\right) +1}C_{s-rj}\right) =C_{s-r\frac{3q-1}{10}}.
\]

We have the following:%
\begin{eqnarray*}
&&\left( \cup _{j=\frac{3}{10}\left( q-7\right) +2}^{\frac{3}{10}\left(
q-7\right) +1+\lambda }C_{s-rj}\right) \cap -q\left( \cup _{j=1}^{\frac{3}{10%
}\left( q-7\right) +1}C_{s-rj}\right) \\
&=&\left( C_{s-r\frac{3q-1}{10}}\cup \left( \cup _{j=\frac{3}{10}\left(
q-7\right) +3}^{\frac{3}{10}\left( q-7\right) +1+\lambda }C_{s-rj}\right)
\right) \cap -q\left( \cup _{j=1}^{\frac{3}{10}\left( q-7\right)
+1}C_{s-rj}\right) \\
&=&\left( C_{s-r\frac{3q-1}{10}}\cap -q\left( \cup _{j=1}^{\frac{3}{10}%
\left( q-7\right) +1}C_{s-rj}\right) \right) \\
&&\cup \left( \left( \cup _{j=\frac{3}{10}\left( q-7\right) +3}^{\frac{3}{10}%
\left( q-7\right) +1+\lambda }C_{s-rj}\right) \cap -q\left( \cup _{j=1}^{%
\frac{3}{10}\left( q-7\right) +1}C_{s-rj}\right) \right) .
\end{eqnarray*}%
We claim that
\[
\left( C_{s-r\frac{3q-1}{10}}\cap -q\left( \cup _{j=1}^{\frac{3}{10}\left(
q-7\right) +1}C_{s-rj}\right) \right) =C_{s-r\frac{3q-1}{10}}
\]%
and
\[
\left( \cup _{j=\frac{3}{10}\left( q-7\right) +3}^{\frac{3}{10}\left(
q-7\right) +1+\lambda }C_{s-rj}\right) \cap -q\left( \cup _{j=1}^{\frac{3}{10%
}\left( q-7\right) +1}C_{s-rj}\right) =\emptyset ,
\]%
where $1\leq \lambda \leq \frac{q+3}{10}.$

Contrary to the our claim, let assume that
\begin{eqnarray*}
&&\left( \cup _{j=\frac{3}{10}\left( q-7\right) +3}^{\frac{3}{10}\left(
q-7\right) +1+\lambda }C_{s-rj}\right) \cap -q\left( \cup _{j=1}^{\frac{3}{10%
}\left( q-7\right) +1}C_{s-rj}\right) \\
&=&\left( \cup _{j=3}^{\lambda }C_{s-r\left( j+\frac{3}{10}\left( q-7\right)
+1\right) }\right) \cap -q\left( \cup _{j=1}^{\frac{3}{10}\left( q-7\right)
+1}C_{s-rj}\right) \neq \emptyset ,
\end{eqnarray*}%
where $1\leq \lambda \leq \frac{q+3}{10}.$ If $\left( \cup _{j=3}^{\lambda
}C_{s-r\left( j+\frac{3}{10}\left( q-7\right) +1\right) }\right) \cap
-q\left( \cup _{j=1}^{\frac{3}{10}\left( q-7\right) +1}C_{s-rj}\right) \neq
\emptyset ,$ then there exists two integers $u^{\prime }$ and $v^{\prime },$
where $3\leq u^{\prime }\leq \frac{q+3}{10},$ and $1\leq v^{\prime }\leq
\frac{3\left( q-7\right) }{10}+1$ such that $s-r\left( u^{\prime }+\frac{3}{%
10}\left( q-7\right) +1\right) \equiv -q\left( s-rv^{\prime }\right) $ $%
\left( mod\text{ }rn\right) $ or $s-r\left( u^{\prime }+\frac{3}{10}\left(
q-7\right) +1\right) \equiv -q\left( s+rv^{\prime }\right) $ $\left( mod%
\text{ }rn\right) .$

\textbf{Case 1:} Let $s-r\left( u^{\prime }+\frac{3}{10}\left( q-7\right)
+1\right) \equiv -q\left( s-rv^{\prime }\right) $ $\left( mod\text{ }%
rn\right) .$ It follows that $s\equiv \left( u^{\prime }+\frac{3}{10}\left(
q-7\right) +1+qv^{\prime }\right) \ \left( mod\text{ }n\right) .$ We know
that $s\equiv 0\ \left( mod\text{ }n\right) .$ Then, $0\leq u^{\prime
}+qv^{\prime }<n.$ This is possible only when $u^{\prime }+qv^{\prime }=n.$
Let $u^{\prime }+qv^{\prime }=\frac{q^{2}+1}{10},$ then $u^{\prime
}+qv^{\prime }=q\frac{\left( q-7\right) }{10}+\frac{7q+1}{10}.$ This
requires that $u^{\prime }=\frac{7q+1}{10},$ which is in contradiction with $%
3\leq u^{\prime }\leq \frac{q+3}{10}.$

\textbf{Case 2:} Let $s-r\left( u^{\prime }+\frac{3}{10}\left( q-7\right)
+1\right) \equiv -q\left( s+rv^{\prime }\right) $ $\left( mod\text{ }%
rn\right) .$ Then, $s\equiv \left( u^{\prime }+\frac{3}{10}\left( q-7\right)
+1-qv^{\prime }\right) \ \left( mod\text{ }n\right) .$ It follows that $%
-3n<u^{\prime }-qv^{\prime }<0,$ This is possible only when $u^{\prime
}-qv^{\prime }=-n$ or $u^{\prime }-qv^{\prime }=-2n.$ Let $u^{\prime
}-qv^{\prime }=-\frac{q^{2}+1}{10},$ then $u^{\prime }-qv^{\prime }=-q\frac{%
\left( q-3\right) }{10}-\frac{3q+1}{10}.$ This requires that $u=-\frac{3q+1}{%
10},$ which contradicts with $2\leq u\leq \frac{q-3}{10}.$ The case $%
u^{\prime }-qv^{\prime }=-2n$ can be shown in a similar way with the case $%
u^{\prime }-qv^{\prime }=-n.$

The above discussions show that
\[
\left( \cup _{j=\frac{3}{10}\left( q-7\right) +2}^{\frac{3}{10}\left(
q-7\right) +1+\lambda }C_{s-rj}\right) \cap -q\left( \cup _{j=1}^{\frac{3}{10%
}\left( q-7\right) +1}C_{s-rj}\right) =C_{s-r\frac{3q-1}{10}}.
\]%
From \ref{lm7*}, we have $-q\left( s-r\frac{\left( q+3\right) }{10}\right)
\equiv s-r\frac{3q-1}{10}$ $\left( mod\text{ }rn\right) .$ This fact says
that

\begin{eqnarray*}
-q\left( \left( \cup _{j=\frac{3}{10}\left( q-7\right) +2}^{\frac{3}{10}%
\left( q-7\right) +1+\lambda }C_{s-rj}\right) \cap -q\left( \cup _{j=1}^{%
\frac{3}{10}\left( q-7\right) +1}C_{s-rj}\right) \right) &=&-qC_{s-r\frac{%
3q-1}{10}} \\
&=&C_{s-r\frac{\left( q+3\right) }{10}},
\end{eqnarray*}%
and it follows that
\[
\left( \cup _{j=1}^{\frac{3}{10}\left( q-7\right) +1}C_{s-rj}\right) \cap
-q\left( \cup _{j=\frac{3}{10}\left( q-7\right) +2}^{\frac{3}{10}\left(
q-7\right) +1+\lambda }C_{s-rj}\right) =C_{s-r\frac{\left( q+3\right) }{10}%
}.
\]

Contrary to the our claim, assume that%
\begin{eqnarray*}
&&\left( \cup _{j=\frac{3}{10}\left( q-7\right) +2}^{\frac{3}{10}\left(
q-7\right) +1+\lambda }C_{s-rj}\right) \cap -q\left( \cup _{j=\frac{3}{10}%
\left( q-7\right) +2}^{\frac{3}{10}\left( q-7\right) +1+\lambda
}C_{s-rj}\right) \\
&=&\left( \cup _{j=2}^{\lambda +1}C_{s-r\left( j+\frac{3}{10}\left(
q-7\right) \right) }\right) \cap -q\left( \cup _{j=1}^{\lambda }C_{s-r\left(
j+\frac{3}{10}\left( q-7\right) \right) }\right) \neq \emptyset ,
\end{eqnarray*}%
where $1\leq \lambda \leq \frac{q+3}{10}.$ If $\left( \cup _{j=2}^{\lambda
+1}C_{s-r\left( j+\frac{3}{10}\left( q-7\right) \right) }\right) \cap
-q\left( \cup _{j=1}^{\lambda }C_{s-r\left( j+\frac{3}{10}\left( q-7\right)
\right) }\right) \neq \emptyset ,$ then there exists two integers $b^{\prime
}$ and $b^{\prime \prime },$ where $1\leq b^{\prime },b^{\prime \prime }\leq
\frac{q+3}{10},$ such that $s-r\left( b^{\prime }+\frac{3}{10}\left(
q-7\right) \right) \equiv -q\left( s-rb^{\prime \prime }\right) $ $\left( mod%
\text{ }rn\right) $ or $s-r\left( b^{\prime }+\frac{3}{10}\left( q-7\right)
\right) \equiv -q\left( s+rb^{\prime \prime }\right) $ $\left( mod\text{ }%
rn\right) .$

\textbf{Case 1:} Let $s-r\left( b^{\prime }+\frac{3}{10}\left( q-7\right)
\right) \equiv -q\left( s-rb^{\prime \prime }\right) $ $\left( mod\text{ }%
rn\right) .$ It follows that $s\equiv \left( b^{\prime }+\frac{3}{10}\left(
q-7\right) +qb^{\prime \prime }\right) \ \left( mod\text{ }n\right) .$ It is
known that $s\equiv 0\ \left( mod\text{ }n\right) .$ Then, we have $2+q\leq
b^{\prime }+qb^{\prime \prime }\leq \frac{\left( q+3\right) }{10}+q\frac{%
\left( q+3\right) }{10}=n+\frac{2q+1}{5}<2n.$ This is possible only when $%
b^{\prime }+qb^{\prime \prime }=n.$ Then, $b^{\prime }+qb^{\prime \prime }=q%
\frac{\left( q+3\right) }{10}+\frac{-3q+1}{10}.$ This means that $b^{\prime
}=\frac{-3q+1}{10},$ which is in contradiction with fact $1\leq b^{\prime
}\leq \frac{q+3}{10}.$

\textbf{Case 2:} Let $s-r\left( b^{\prime }+\frac{3}{10}\left( q-7\right)
\right) \equiv -q\left( s+rb^{\prime \prime }\right) $ $\left( mod\text{ }%
rn\right) .$ Then, $s\equiv \left( b^{\prime }+\frac{3}{10}\left( q-7\right)
-qb^{\prime \prime }\right) \ \left( mod\text{ }n\right) .$ It follows that $%
-2n<-\frac{\left( q^{2}+1\right) }{10}-\frac{\left( 3q-11\right) }{10}\leq
b^{\prime }+qb^{\prime \prime }\leq \frac{-9q+3}{10}<0.$ This is possible
only when $b^{\prime }+qb^{\prime \prime }=-n.$ Let $b^{\prime }+qb^{\prime
\prime }=-\frac{q^{2}+1}{10},$ then $b^{\prime }+qb^{\prime \prime }=-q\frac{%
\left( q+3\right) }{10}-\frac{3q+1}{10}.$ This requires that $b^{\prime }=%
\frac{3q-1}{10},$ which contradicts with $1\leq b^{\prime }\leq \frac{q+3}{10%
}.$

Consequently, we have%
\[
\left( \cup _{j=\frac{3}{10}\left( q-7\right) +2}^{\frac{3}{10}\left(
q-7\right) +1+\lambda }C_{s-rj}\right) \cap -q\left( \cup _{j=\frac{3}{10}%
\left( q-7\right) +2}^{\frac{3}{10}\left( q-7\right) +1+\lambda
}C_{s-rj}\right) =\emptyset .
\]

From Lemma \ref{lm1},\ we have $c=5,$ and by Theorem \ref{th1}, there exists
EAQMDS codes with parameters
\[
\llbracket n,n-\frac{6}{5}\left( q-7\right) -4\lambda -1,\frac{3}{5}\left(
q-7\right) +2\lambda +4;5\rrbracket_{q},
\]%
where $1\leq \lambda \leq \frac{q+3}{10}.$
\end{proof}

\begin{example}
We present some parameters of EAQMDS codes obtained from Theorem \ref{th4*}
in Table \ref{table3*}.
\end{example}

\begin{table}[tbp]
\caption{Some EAQMDS codes obtained by Theorem \protect\ref{th3*}.}
\label{table3*}\centering
\begin{tabular}{cc}
\hline
$q$ & $\llbracket n,n-\frac{6}{5}\left( q-7\right) -4\lambda -1,\frac{3}{5}%
\left( q-7\right) +2\lambda +4;5\rrbracket_{q}$ \\ \hline
$17$ & $\llbracket17,4,12{;5}\rrbracket_{17}$ \\
$17$ & $\llbracket29,8,14{;5}\rrbracket_{17}$ \\
$37$ & $\llbracket137,96,24{;5}\rrbracket_{37}$ \\
$37$ & $\llbracket137,92,26{;5}\rrbracket_{37}$ \\
$37$ & $\llbracket137,88,28{;5}\rrbracket_{37}$ \\
$47$ & $\llbracket221,168,30{;5}\rrbracket_{47}$ \\
$47$ & $\llbracket221,164,32{;5}\rrbracket_{47}$ \\
$47$ & $\llbracket221,160,34{;5}\rrbracket_{47}$ \\
$47$ & $\llbracket221,156,36{;5}\rrbracket_{47}$ \\ \hline
\end{tabular}%
\end{table}

Let $Z$ and $\overline{Z}$ be the sets defined in Lemma \ref{lm5*} and Lemma %
\ref{lm6*} respectively. Define $T=Z\cup \overline{Z}\cup F=\cup _{j=0}^{%
\frac{2q+1}{5}+\lambda },$ where $F=\cup _{j=\frac{2q+1}{5}+1}^{\frac{2q+1}{5%
}+1+\lambda }$ and $1\leq \lambda \leq \frac{q+3}{10}.$ By combining results
of Lemma \ref{lm5*}, Lemma \ref{lm6*}, Lemma \ref{lm7*} and Theorem \ref%
{th4*} we see that the number of entangled states $c=9.$ Based on this fact,
in Theorem \ref{th5*}, we give a class of EAQMDS codes of length $n=\frac{%
q^{2}+1}{10}$ and with entangled states $c=9.$

\begin{theorem}
\label{th5*} Let $q\equiv 7\left( mod\text{ }10\right) .$ If $\mathcal{C}$
is an $q^{2}$-ary $\alpha $-constacyclic code of length $n$ with defining
set $T=\cup _{j=0}^{\frac{2q+1}{5}+\lambda },$ then there exists EAQMDS
codes with parameters
\[
\llbracket n,n-\frac{4}{5}\left( 2q+1\right) -4\lambda +7,\frac{2}{5}\left(
2q+1\right) +2\lambda +2;9\rrbracket_{q},
\]%
where $1\leq \lambda \leq \frac{\left( q+3\right) }{10}.$
\end{theorem}

\begin{proof}
The proof is a direct result of Lemma \ref{lm5*}, Lemma \ref{lm6*}, Lemma %
\ref{lm7*} and Theorem \ref{th4*}.
\end{proof}

\begin{example}
We present some parameters of EAQMDS codes obtained from Theorem \ref{th5*}
in Table \ref{table4*}.
\end{example}

\begin{table}[tbp]
\caption{Some EAQMDS codes obtained by Theorem \protect\ref{th2*}.}
\label{table4*}\centering
\begin{tabular}{cc}
\hline
$q$ & $\llbracket n,n-\frac{4}{5}\left( 2q+1\right) -4\lambda +7,\frac{2}{5}%
\left( 2q+1\right) +2\lambda +2;9\rrbracket_{q}$ \\ \hline
$17$ & $\llbracket29,4,18{;9}\rrbracket_{17}$ \\
$37$ & $\llbracket137,80,34{;9}\rrbracket_{37}$ \\
$37$ & $\llbracket137,76,36{;9}\rrbracket_{37}$ \\
$37$ & $\llbracket137,72,38{;9}\rrbracket_{37}$ \\
$47$ & $\llbracket221,148,42{;9}\rrbracket_{47}$ \\
$47$ & $\llbracket221,144,44{;9}\rrbracket_{47}$ \\
$47$ & $\llbracket221,140,46{;9}\rrbracket_{47}$ \\
$47$ & $\llbracket221,136,48{;9}\rrbracket_{47}$ \\ \hline
\end{tabular}%
\end{table}

\subsection{EAQMDS codes of length $n=\frac{q^{2}+1}{10},$ where $q\equiv
3\left( mod\text{ }10\right) $}

Note that $n=\frac{q^{2}+1}{10},$ and so $ord_{rn}\left( q^{2}\right) =2.$
This means that each $q^{2}$-cyclotomic coset modulo $rn$ includes one or
two elements. Let $q\equiv 3\left( mod\text{ }10\right) $ and $s=\frac{%
q^{2}+1}{2}.$ It is easy to see that the $q^{2}$-cyclotomic cosets modulo $%
rn $ containing some integers from $1$ to $rn$ are $C_{s}=\left\{ s\right\} $
and $C_{s-rj}=\left\{ s-rj,s+rj\right\} ,$ where $1\leq j\leq \frac{q-1}{2}.$

\begin{lemma}
\label{lm2*} Let $q\equiv 3\left( mod\text{ }10\right) .$ If $\mathcal{C}$
is a $q^{2}$-ary constacyclic code of length $n$ and defining set $Z=\cup
_{j=1}^{\lambda }C_{s-rj},$ where $1\leq \lambda \leq \frac{3\left(
q-3\right) }{10},$ then $\mathcal{C}^{\perp _{H}}\subseteq \mathcal{C}.$
\end{lemma}

\begin{proof}
By Proposition \ref{prop2}, it is sufficient to prove that $Z\cap \left(
-qZ\right) =\emptyset .$ Suppose that $Z\cap \left( -qZ\right) \neq
\emptyset .$ Then, there exists two integers $j,k,$ where $1\leq j,k\leq
\frac{3\left( q-3\right) }{10},$ such that $s-rj\equiv -q\left( s-rk\right)
\left( mod\text{ }rn\right) $ or $s-rj\equiv -q\left( s+rk\right) \left( mod%
\text{ }rn\right) .$

\textbf{Case 1:} Let $s-rj\equiv -q\left( s-rk\right) \left( mod\text{ }%
rn\right) .$ This is equivalent to $s\equiv j+qk\left( mod\text{ }n\right) .$
As $s\equiv 0\left( mod\text{ }n\right) ,$ we get $j+qk\equiv 0\left( mod%
\text{ }n\right) .$ Since $1\leq j,k\leq \frac{3q-9}{10},$ $q+1\leq j+qk\leq
\frac{3q-9}{10}+q\frac{3q-9}{10}=\left( q+1\right) \frac{3\left( q-3\right)
}{10}<3n.$ Then, we have that $j+qk\equiv n\left( mod\text{ }n\right) $ or $%
j+qk\equiv 2n\left( mod\text{ }n\right) .$ If $j+qk=n,$ then $j+qk=\frac{%
q^{2}+1}{10}=q\frac{\left( q-3\right) }{10}+\frac{3q+1}{10}.$ By division
algorithm, $j=\frac{3q+1}{10}.$ This is a contradiction, because $0\leq
j\leq \frac{3\left( q-3\right) }{10}.$ If $j+qk=2n,$ then $j+qk=\frac{%
2\left( q^{2}+1\right) }{10}=q\frac{2\left( q-3\right) }{10}+\frac{2\left(
3q+1\right) }{10}.$ By division algorithm, $j=\frac{2\left( 3q+1\right) }{10}%
.$ This is a contradiction, because $0\leq j\leq \frac{3\left( q-3\right) }{%
10}.$

\textbf{Case 2:} Let $s-rj\equiv -q\left( s+rk\right) \left( mod\text{ }%
rn\right) .$ This is equivalent to $s\equiv j-qk\left( mod\text{ }n\right) .$
By $s\equiv 0\left( mod\text{ }n\right) ,$ we have $j-qk\equiv 0\left( mod%
\text{ }n\right) .$ Since $1\leq j,k\leq \frac{3\left( q-3\right) }{10},$ $%
-4n<1-3q\frac{\left( q-3\right) }{10}\leq j-qk\leq \frac{-7q-9}{10}<-n.$ We
have that $j-qk\equiv -2n\left( mod\text{ }n\right) $ or $j-qk\equiv
-3n\left( mod\text{ }n\right) .$ If $j+qk=-2n,$ then $j-qk=\frac{-2\left(
q^{2}+1\right) }{10}=-q\frac{2\left( q-3\right) }{10}+\frac{2\left(
3q-1\right) }{10}.$ By division algorithm, $k=\frac{2\left( q-3\right) }{10}%
. $ This is a contradiction, because $0\leq k\leq \frac{3\left( q-3\right) }{%
10}.$ If $j+qk=-3n,$ then $j-qk=\frac{3\left( q^{2}+1\right) }{10}=-q\frac{%
3\left( q-3\right) }{10}+\frac{3\left( 3q-1\right) }{10}.$ By division
algorithm, $k=\frac{3\left( q-3\right) }{10}.$ This is a contradiction,
because $0\leq k\leq \frac{3\left( q-3\right) }{10}.$
\end{proof}

\begin{lemma}
\label{lm4*} Let $q\equiv 3\left( mod\text{ }10\right) $ and $s=\frac{q^{2}+1%
}{2}.$ Then we have the following:

\begin{enumerate}
\item $-qC_{s}=C_{s}=\left\{ s\right\} ,$

\item $-qC_{s-r\frac{\left( q-3\right) }{10}}=C_{s-r\left( \frac{3q+1}{10}%
\right) }=\left\{ s-r\left( \frac{3q+1}{10}\right) ,s+r\left( \frac{3q+1}{10}%
\right) \right\} ,$

\item $-qC_{s-r\left( \frac{q+2}{5}\right) }=C_{s-r\left( \frac{2q-1}{5}%
\right) }=\left\{ s-r\frac{\left( 2q-1\right) }{5},s+r\frac{\left(
2q-1\right) }{5}\right\} .$
\end{enumerate}
\end{lemma}

\begin{lemma}
\label{lm3*} Let $q\equiv 3\left( mod\text{ }10\right) .$ If $\mathcal{C}$
is a $q^{2}$-ary constacyclic code of length $n$ and defining set $\overline{%
Z}=\cup _{j=\frac{3\left( q-3\right) }{10}}^{t}C_{s-rj},$ where $1\leq t\leq
\frac{\left( q-3\right) }{10},$ then $\mathcal{C}^{\perp _{H}}\subseteq
\mathcal{C}.$
\end{lemma}

\begin{proof}
The proof for this lemma is very similar to the proof of Lemma \ref{lm2*}.
\end{proof}

As an immediate result of Lemma \ref{lm3*} we have $\overline{Z}\cap \left(
-q\overline{Z}\right) =\emptyset .$

In Theorem \ref{th2*}, we give a class of EAQMDS codes of length $n=\frac{%
q^{2}+1}{10}$ and with entangled states $c=5.$

\begin{theorem}
\label{th2*} Let $q\equiv 3\left( mod\text{ }10\right) .$ If $\mathcal{C}$
is an $q^{2}$-ary $\alpha $-constacyclic code of length $n$ with defining
set $Z=\cup _{j=0}^{\frac{3\left( q-3\right) }{10}+\lambda }C_{s-rj},$ then
there exists EAQMDS codes with parameters
\[
\llbracket n,n-\frac{6}{5}\left( q-3\right) -4\lambda +3,\frac{3}{5}\left(
q-3\right) +2\lambda +2;5\rrbracket_{q},
\]%
where $1\leq \lambda \leq \frac{\left( q-3\right) }{10}.$
\end{theorem}

\begin{proof}
Since the defining set of an $\alpha $-constacyclic code $\mathcal{C}$ of
length $n$ is $Z=\cup _{j=0}^{\frac{3}{10}\left( q-3\right) +\lambda
}C_{s-rj},$ and the cardinality of $Z,$ which we denoted by $\left\vert
Z\right\vert ,$ is $\left\vert Z\right\vert =\frac{3}{5}\left( q-3\right)
+2\lambda +1,$ then by Proposition \ref{prop1} and \ref{prop3}, $\mathcal{C}$
is a $q^{2}$-ary MDS $\alpha $-constacyclic code with parameters
\[
\left[ n,n-\left( \frac{3}{5}\left( q-3\right) +2\lambda +1\right) ,\frac{3}{%
5}\left( q-3\right) +2\lambda +2\right] _{q^{2}}.
\]%
Hence, we have the following:%
\begin{eqnarray*}
Z_{1} &=&Z\cap (-qZ)= \\
&&\left( \left( \cup _{j=0}^{\frac{3}{10}\left( q-3\right) }C_{s-rj}\right)
\cup \left( \cup _{j=\frac{3}{10}\left( q-3\right) +1}^{\frac{3}{10}\left(
q-3\right) +\lambda }C_{s-rj}\right) \right) \cap \\
&&\left( -q\left( \cup _{j=0}^{\frac{3}{10}\left( q-3\right)
}C_{s-rj}\right) \cup -q\left( \cup _{j=\frac{3}{10}\left( q-3\right) +1}^{%
\frac{3}{10}\left( q-3\right) +\lambda }C_{s-rj}\right) \right) \\
&=&\left( \left( \cup _{j=0}^{\frac{3}{10}\left( q-3\right) }C_{s-rj}\right)
\cap -q\left( \cup _{j=0}^{\frac{3}{10}\left( q-3\right) }C_{s-rj}\right)
\right) \cup \\
&&\left( \left( \cup _{j=0}^{\frac{3}{10}\left( q-3\right) }C_{s-rj}\right)
\cap -q\left( \cup _{j=\frac{3}{10}\left( q-3\right) +1}^{\frac{3}{10}\left(
q-3\right) +\lambda }C_{s-rj}\right) \right) \\
&&\cup \left( \left( \cup _{j=\frac{3}{10}\left( q-3\right) +1}^{\frac{3}{10}%
\left( q-3\right) +\lambda }C_{s-rj}\right) \cap -q\left( \cup _{j=0}^{\frac{%
3}{10}\left( q-3\right) }C_{s-rj}\right) \right) \\
&&\cup \left( \left( \cup _{j=\frac{3}{10}\left( q-3\right) +1}^{\frac{3}{10}%
\left( q-3\right) +\lambda }C_{s-rj}\right) \cap -q\left( \cup _{j=\frac{3}{%
10}\left( q-3\right) +1}^{\frac{3}{10}\left( q-3\right) +\lambda
}C_{s-rj}\right) \right) .
\end{eqnarray*}%
We claim that
\[
Z_{1}=Z\cap (-qZ)=C_{s}\cup C_{s-r\frac{3q+1}{10}}\cup C_{s-r\frac{\left(
q-3\right) }{10}}.
\]%
From Lemma \ref{lm2}, we have $\left( \cup _{j=1}^{\frac{3}{10}\left(
q-3\right) }C_{s-rj}\right) \cap -q\left( \cup _{j=1}^{\frac{3}{10}\left(
q-3\right) }C_{s-rj}\right) =\emptyset .$ By examining the coset structure
of the defining set $Z,$ we can see that if $j=0,$ then $\ C_{s}\cap
-qC_{s}=\left\{ s\right\} .$

Thus, we need to show that%
\[
\left( \cup _{j=1}^{\frac{3}{10}\left( q-3\right) }C_{s-rj}\right) \cap
-q\left( \cup _{j=\frac{3}{10}\left( q-3\right) +1}^{\frac{3}{10}\left(
q-3\right) +\lambda }C_{s-rj}\right) =C_{s-r\frac{\left( q-3\right) }{10}},
\]%
\[
\left( \cup _{j=\frac{3}{10}\left( q-3\right) +1}^{\frac{3}{10}\left(
q-3\right) +\lambda }C_{s-rj}\right) \cap -q\left( \cup _{j=1}^{\frac{3}{10}%
\left( q-3\right) }C_{s-rj}\right) =C_{s-r\frac{3q+1}{10}},
\]%
\[
\left( \cup _{j=\frac{3}{10}\left( q-3\right) +1}^{\frac{3}{10}\left(
q-3\right) +\lambda }C_{s-rj}\right) \cap -q\left( \cup _{j=\frac{3}{10}%
\left( q-3\right) +1}^{\frac{3}{10}\left( q-3\right) +\lambda
}C_{s-rj}\right) =\emptyset .
\]

We first show that
\[
\left( \cup _{j=\frac{3}{10}\left( q-3\right) +1}^{\frac{3}{10}\left(
q-3\right) +\lambda }C_{s-rj}\right) \cap -q\left( \cup _{j=1}^{\frac{3}{10}%
\left( q-3\right) }C_{s-rj}\right) =C_{s-r\frac{3q+1}{10}}.
\]

We have the following:%
\begin{eqnarray*}
&&\left( \cup _{j=\frac{3}{10}\left( q-3\right) +1}^{\frac{3}{10}\left(
q-3\right) +\lambda }C_{s-rj}\right) \cap -q\left( \cup _{j=1}^{\frac{3}{10}%
\left( q-3\right) }C_{s-rj}\right) \\
&=&\left( C_{s-r\frac{3q+1}{10}}\cup \left( \cup _{j=\frac{3}{10}\left(
q-3\right) +2}^{\frac{3}{10}\left( q-3\right) +\lambda }C_{s-rj}\right)
\right) \cap -q\left( \cup _{j=1}^{\frac{3}{10}\left( q-3\right)
}C_{s-rj}\right) \\
&=&\left( C_{s-r\frac{3q+1}{10}}\cap -q\left( \cup _{j=1}^{\frac{3}{10}%
\left( q-3\right) }C_{s-rj}\right) \right) \\
&&\cup \left( \left( \cup _{j=\frac{3}{10}\left( q-3\right) +2}^{\frac{3}{10}%
\left( q-3\right) +\lambda }C_{s-rj}\right) \cap -q\left( \cup _{j=1}^{\frac{%
3}{10}\left( q-3\right) }C_{s-rj}\right) \right) \\
&=&C_{s-r\frac{3q+1}{10}}.
\end{eqnarray*}%
We claim that
\begin{eqnarray*}
\left( C_{s-r\frac{3q+1}{10}}\cap -q\left( \cup _{j=1}^{\frac{3}{10}\left(
q-3\right) }C_{s-rj}\right) \right) &=&C_{s-r\frac{3q+1}{10}}, \\
\left( \cup _{j=\frac{3}{10}\left( q-3\right) +2}^{\frac{3}{10}\left(
q-3\right) +\lambda }C_{s-rj}\right) \cap -q\left( \cup _{j=1}^{\frac{3}{10}%
\left( q-3\right) }C_{s-rj}\right) &=&\emptyset
\end{eqnarray*}%
where $1\leq \lambda \leq \frac{q-3}{10}.$

Contrary to the claim, assume that
\begin{eqnarray*}
&&\left( \cup _{j=\frac{3}{10}\left( q-3\right) +2}^{\frac{3}{10}\left(
q-3\right) +\lambda }C_{s-rj}\right) \cap -q\left( \cup _{j=1}^{\frac{3}{10}%
\left( q-3\right) }C_{s-rj}\right) \\
&=&\left( \cup _{j=2}^{\lambda }C_{s-r\left( j+\frac{3}{10}\left( q-3\right)
\right) }\right) \cap -q\left( \cup _{j=1}^{\frac{3}{10}\left( q-3\right)
}C_{s-rj}\right) \neq \emptyset ,
\end{eqnarray*}%
where $1\leq \lambda \leq \frac{q-3}{10}.$ If $\left( \cup _{j=2}^{\lambda
}C_{s-r\left( j+\frac{3}{10}\left( q-3\right) \right) }\right) \cap -q\left(
\cup _{j=0}^{\frac{3}{10}\left( q-3\right) }C_{s-rj}\right) \neq \emptyset ,$
then there exists two integers $u$ and $v,$ where $2\leq u\leq \frac{q-3}{10}%
,$ and $1\leq v\leq \frac{3\left( q-3\right) }{10}$ such that $s-r\left( u+%
\frac{3}{10}\left( q-3\right) \right) \equiv -q\left( s-rv\right) $ $\left(
mod\text{ }rn\right) $ or $s-r\left( u+\frac{3}{10}\left( q-3\right) \right)
\equiv -q\left( s+rv\right) $ $\left( mod\text{ }rn\right) .$

\textbf{Case 1:} Let $s-r\left( u+\frac{3}{10}\left( q-3\right) \right)
\equiv -q\left( s-rv\right) $ $\left( mod\text{ }rn\right) .$ It follows
that $s\equiv \left( u+\frac{3}{10}\left( q-3\right) +qv\right) \ \left( mod%
\text{ }n\right) .$ We know that $s\equiv 0\ \left( mod\text{ }n\right) .$
Then, $2+q\leq u+qv\leq \frac{\left( q-3\right) }{10}+3q\frac{\left(
q-3\right) }{10}=\frac{3\left( q^{2}+1\right) }{10}-\frac{\left( 8q+6\right)
}{10}<3n.$ This is possible only when $u+qv=n$ or $u+qv=2n.$ Let $u+qv=\frac{%
q^{2}+1}{10},$ then $u+qv=q\frac{\left( q-3\right) }{10}+\frac{3q+1}{10}.$
This requires that $u=\frac{3q+1}{10},$ which is in contradiction with $%
2\leq u\leq \frac{q-3}{10}.$ Let $u+qv=2\frac{q^{2}+1}{10},$ then $u+qv=2q%
\frac{\left( q-3\right) }{10}+2\frac{3q+1}{10}.$ This requires that $u=\frac{%
3q+1}{10},$ which is in contradiction with $2\leq u\leq \frac{q-3}{10}.$

\textbf{Case 2:} Let $s-r\left( u+\frac{3}{10}\left( q-3\right) \right)
\equiv -q\left( s+rv\right) $ $\left( mod\text{ }rn\right) .$ Then,
\[
s\equiv \left( u+\frac{3}{10}\left( q-3\right) -qv\right) \ \left( mod\text{
}n\right) .
\]
It follows that $-3n<-\frac{\left( 3q^{2}+3\right) }{10}+\frac{9q+23}{10}%
\leq u-qv\leq \frac{q-3}{10}<0,$ This is possible only when $u-qv=-n$ or $%
u-qv=-2n.$ Let $u-qv=-\frac{q^{2}+1}{10},$ then $u-qv=-q\frac{\left(
q-3\right) }{10}-\frac{3q+1}{10}.$ This requires that $u=-\frac{3q+1}{10},$
which contradicts with $2\leq u\leq \frac{q-3}{10}.$ Let $u-qv=-2\frac{%
q^{2}+1}{10},$ then $u-qv=-2q\frac{\left( q-3\right) }{10}-\frac{2\left(
3q+1\right) }{10}.$ This requires that $u=-\frac{2\left( 3q+1\right) }{10},$
which is in contradiction with $2\leq u\leq \frac{q-3}{10}.$

The above discussions show that
\[
\left( \cup _{j=\frac{3}{10}\left( q-3\right) +1}^{\frac{3}{10}\left(
q-3\right) +\lambda }C_{s-rj}\right) \cap -q\left( \cup _{j=1}^{\frac{3}{10}%
\left( q-3\right) }C_{s-rj}\right) =C_{s-r\frac{3q+1}{10}}.
\]%
We have $-q\left( s-r\frac{3q+1}{10}\right) \equiv s-r\frac{\left(
q-3\right) }{10}$ $\left( mod\text{ }rn\right) .$ This fact says that

\[
-q\left( \left( \cup _{j=\frac{3}{10}\left( q-3\right) +1}^{\frac{3}{10}%
\left( q-3\right) +\lambda }C_{s-rj}\right) \cap -q\left( \cup _{j=1}^{\frac{%
3}{10}\left( q-3\right) }C_{s-rj}\right) \right) =-qC_{s-r\frac{3q+1}{10}%
}=C_{s-r\frac{\left( q-3\right) }{10}},
\]%
and it follows that
\[
\left( \cup _{j=1}^{\frac{3}{10}\left( q-3\right) }C_{s-rj}\right) \cap
-q\left( \cup _{j=\frac{3}{10}\left( q-3\right) +1}^{\frac{3}{10}\left(
q-3\right) +\lambda }C_{s-rj}\right) =C_{s-r\frac{\left( q-3\right) }{10}}.
\]

Contrary to the our claim, suppose that%
\begin{eqnarray*}
&&\left( \cup _{j=\frac{3}{10}\left( q-3\right) +1}^{\frac{3}{10}\left(
q-3\right) +\lambda }C_{s-rj}\right) \cap -q\left( \cup _{j=\frac{3}{10}%
\left( q-3\right) +1}^{\frac{3}{10}\left( q-3\right) +\lambda
}C_{s-rj}\right) \\
&=&\left( \cup _{j=2}^{\lambda }C_{s-r\left( j+\frac{3}{10}\left( q-3\right)
\right) }\right) \cap -q\left( \cup _{j=1}^{\lambda }C_{s-r\left( j+\frac{3}{%
10}\left( q-3\right) \right) }\right) \neq \emptyset ,
\end{eqnarray*}%
where $1\leq \lambda \leq \frac{q-3}{10}.$ If $\left( \cup _{j=2}^{\lambda
}C_{s-r\left( j+\frac{3}{10}\left( q-3\right) \right) }\right) \cap -q\left(
\cup _{j=1}^{\lambda }C_{s-r\left( j+\frac{3}{10}\left( q-3\right) \right)
}\right) \neq \emptyset ,$ then there exists two integers $a^{\prime }$ and $%
a^{\prime \prime },$ where $1\leq a^{\prime },a^{\prime \prime }\leq \frac{%
q-3}{10},$ such that $s-r\left( a^{\prime }+\frac{3}{10}\left( q-3\right)
\right) \equiv -q\left( s-ra^{\prime \prime }\right) $ $\left( mod\text{ }%
rn\right) $ or $s-r\left( a^{\prime }+\frac{3}{10}\left( q-3\right) \right)
\equiv -q\left( s+ra^{\prime \prime }\right) $ $\left( mod\text{ }rn\right)
. $

\textbf{Case 1:} Let $s-r\left( a^{\prime }+\frac{3}{10}\left( q-3\right)
\right) \equiv -q\left( s-ra^{\prime \prime }\right) $ $\left( mod\text{ }%
rn\right) .$ It follows that $s\equiv \left( a^{\prime }+\frac{3}{10}\left(
q-3\right) +qa^{\prime \prime }\right) \ \left( mod\text{ }n\right) .$ We
know that $s\equiv 0\ \left( mod\text{ }n\right) .$ Then, $2+q\leq a^{\prime
}+qa^{\prime \prime }\leq \frac{\left( q-3\right) }{10}+3q\frac{\left(
q-3\right) }{10}=\frac{3\left( q^{2}+1\right) }{10}-\frac{\left( 8q+6\right)
}{10}<3n.$ This is possible only when $a^{\prime }+qa^{\prime \prime }=n$ or
$a^{\prime }+qa^{\prime \prime }=2n.$ Let $a^{\prime }+qa^{\prime \prime }=%
\frac{q^{2}+1}{10},$ then $a^{\prime }+qa^{\prime \prime }=q\frac{\left(
q-3\right) }{10}+\frac{3q+1}{10}.$ This requires that $a^{\prime }=\frac{3q+1%
}{10},$ which is in contradiction with $2\leq a^{\prime }\leq \frac{q-3}{10}%
. $ Let $a^{\prime }+qa^{\prime \prime }=2\frac{q^{2}+1}{10},$ then $%
a^{\prime }+qa^{\prime \prime }=2q\frac{\left( q-3\right) }{10}+2\frac{3q+1}{%
10}.$ This requires that $a^{\prime }=\frac{3q+1}{10},$ which is in
contradiction with $2\leq a^{\prime }\leq \frac{q-3}{10}.$

\textbf{Case 2:} Let $s-r\left( a^{\prime }+\frac{3}{10}\left( q-3\right)
\right) \equiv -q\left( s+ra^{\prime \prime }\right) $ $\left( mod\text{ }%
rn\right) .$ Then, $s\equiv \left( a^{\prime }+\frac{3}{10}\left( q-3\right)
-qa^{\prime \prime }\right) \ \left( mod\text{ }n\right) .$ It follows that $%
-3n<-\frac{\left( 3q^{2}+3\right) }{10}+\frac{9q+23}{10}\leq a^{\prime
}-qa^{\prime \prime }\leq \frac{q-3}{10}<0,$ This is possible only when $%
a^{\prime }-qa^{\prime \prime }=-n$ or $a^{\prime }-qa^{\prime \prime }=-2n.$
Let $a^{\prime }-qa^{\prime \prime }=-\frac{q^{2}+1}{10},$ then $a^{\prime
}-qa^{\prime \prime }=-q\frac{\left( q-3\right) }{10}-\frac{3q+1}{10}.$ This
requires that $a^{\prime }=-\frac{3q+1}{10},$ which contradicts with $2\leq
a^{\prime }\leq \frac{q-3}{10}.$ Let $a^{\prime }-qa^{\prime \prime }=-2%
\frac{q^{2}+1}{10},$ then $a^{\prime }-qa^{\prime \prime }=-2q\frac{\left(
q-3\right) }{10}-\frac{2\left( 3q+1\right) }{10}.$ This requires that $%
a^{\prime }=-\frac{2\left( 3q+1\right) }{10},$ which is in contradiction
with $2\leq a^{\prime }\leq \frac{q-3}{10}.$

This means that%
\[
\left( \cup _{j=\frac{3}{10}\left( q-3\right) +1}^{\frac{3}{10}\left(
q-3\right) +\lambda }C_{s-rj}\right) \cap -q\left( \cup _{j=\frac{3}{10}%
\left( q-3\right) +1}^{\frac{3}{10}\left( q-3\right) +\lambda
}C_{s-rj}\right) =\emptyset .
\]

From Lemma \ref{lm1},\ we have $c=5,$ and by Theorem \ref{th1}, there exists
EAQMDS codes with parameters
\[
\llbracket n,n-\frac{6}{5}\left( q-3\right) -4\lambda +3,\frac{3}{5}\left(
q-3\right) +2\lambda +2;5\rrbracket_{q},
\]%
where $1\leq \lambda \leq \frac{q-3}{10}.$
\end{proof}

\begin{example}
We present some parameters of EAQMDS codes obtained from Theorem \ref{th2*}
in Table \ref{table1*}.
\end{example}

\begin{table}[tbp]
\caption{Some EAQMDS codes obtained by Theorem \protect\ref{th2*}.}
\label{table1*}\centering
\begin{tabular}{cc}
\hline
$q$ & $\llbracket n,n-\frac{6}{5}\left( q-3\right) -4\lambda +3,\frac{3}{5}%
\left( q-3\right) +2\lambda +2;5\rrbracket_{q}$ \\ \hline
$13$ & $\llbracket17,4,10{;5}\rrbracket_{13}$ \\
$23$ & $\llbracket53,28,16{;5}\rrbracket_{23}$ \\
$23$ & $\llbracket53,24,18{;5}\rrbracket_{23}$ \\
$43$ & $\llbracket185,136,28{;5}\rrbracket_{43}$ \\
$43$ & $\llbracket185,132,30{;5}\rrbracket_{43}$ \\
$43$ & $\llbracket185,128,32{;5}\rrbracket_{43}$ \\
$43$ & $\llbracket185,124,34{;5}\rrbracket_{43}$ \\
$53$ & $\llbracket281,220,34{;5}\rrbracket_{53}$ \\
$53$ & $\llbracket281,216,36{;5}\rrbracket_{53}$ \\
$53$ & $\llbracket281,212,38{;5}\rrbracket_{53}$ \\
$53$ & $\llbracket281,208,40{;5}\rrbracket_{53}$ \\
$53$ & $\llbracket281,204,42{;5}\rrbracket_{53}$ \\ \hline
\end{tabular}%
\end{table}

Let $Z$ and $\overline{Z}$ be the sets defined in Lemma \ref{lm2*} and Lemma %
\ref{lm3*} respectively. Define $T=Z\cup \overline{Z}\cup F=\cup _{j=0}^{%
\frac{4\left( q-3\right) }{10}+\lambda },$ where $F=\cup _{j=\frac{3\left(
q-3\right) }{10}+1}^{\frac{3\left( q-3\right) }{10}+1+\lambda }$ and $1\leq
\lambda \leq \frac{q-3}{10}.$ By combining results of Lemma \ref{lm2*},
Lemma \ref{lm3*}, Lemma \ref{lm4*} and Theorem \ref{th2*} we see that the
number of entangled states $c=9.$ Based on this fact, in Theorem \ref{th3*},
we give a class of EAQMDS codes of length $n=\frac{q^{2}+1}{10}$ and with
entangled states $c=9.$

\begin{theorem}
\label{th3*} Let $q\equiv 3\left( mod\text{ }10\right) .$ If $\mathcal{C}$
is an $q^{2}$-ary $\alpha $-constacyclic code of length $n$ with defining
set $T=Z\cup \overline{Z}=\cup _{j=0}^{\frac{4\left( q-3\right) }{10}%
+\lambda },$ then there exists EAQMDS codes with parameters
\[
\llbracket n,n-\frac{8}{5}\left( q-3\right) -4\lambda +7,\frac{4}{5}\left(
q-3\right) +2\lambda +2;9\rrbracket_{q},
\]%
where $1\leq \lambda \leq \frac{\left( q-3\right) }{10}.$
\end{theorem}

\begin{proof}
Since the defining set of an $\alpha $-constacyclic code $\mathcal{C}$ of
length $n$ is $T=\cup _{j=0}^{\frac{4\left( q-3\right) }{10}+\lambda }$ and
the cardinality of $T,$ which we denoted by $\left\vert T\right\vert ,$ is $%
\left\vert T\right\vert =\frac{4}{5}\left( q-3\right) +2\lambda +1,$ then by
Proposition \ref{prop1} and \ref{prop3}, $\mathcal{C}$ is a $q^{2}$-ary MDS $%
\alpha $-constacyclic code with parameters
\[
\left[ n,n-\left( \frac{4}{5}\left( q-3\right) +2\lambda +1\right) ,\frac{4}{%
5}\left( q-3\right) +2\lambda +2\right] _{q^{2}}.
\]%
The remaining part of the proof is obtained directly from Theorem \ref{th2*}%
, Lemma \ref{lm3*} and Lemma \ref{lm4*}..
\end{proof}

\begin{example}
We present some parameters of EAQMDS codes obtained from Theorem \ref{th3*}
in Table \ref{table2*}.
\end{example}

\begin{table}[tbp]
\caption{Some EAQMDS codes obtained by Theorem \protect\ref{th2*}.}
\label{table2*}\centering
\begin{tabular}{cc}
\hline
$q$ & $\llbracket n,n-\frac{8}{5}\left( q-3\right) -4\lambda +7,\frac{4}{5}%
\left( q-3\right) +2\lambda +2;9\rrbracket_{q}$ \\ \hline
$13$ & $\llbracket17,4,12{;9}\rrbracket_{13}$ \\
$23$ & $\llbracket53,24,20{;9}\rrbracket_{23}$ \\
$23$ & $\llbracket53,20,22{;9}\rrbracket_{23}$ \\
$43$ & $\llbracket185,124,36{;9}\rrbracket_{43}$ \\
$43$ & $\llbracket185,120,38{;9}\rrbracket_{43}$ \\
$43$ & $\llbracket185,116,40{;9}\rrbracket_{43}$ \\
$43$ & $\llbracket185,112,42{;9}\rrbracket_{43}$ \\
$53$ & $\llbracket281,204,44{;9}\rrbracket_{53}$ \\
$53$ & $\llbracket281,200,46{;9}\rrbracket_{53}$ \\
$53$ & $\llbracket281,196,48{;9}\rrbracket_{53}$ \\
$53$ & $\llbracket281,192,50{;9}\rrbracket_{53}$ \\
$53$ & $\llbracket281,188,52{;9}\rrbracket_{53}$ \\ \hline
\end{tabular}%
\end{table}

\section{Construction of EAQMDS codes from constacyclic codes ($q$ is even)}

Throughout this section, we set $q=2^{e},$ $r=q+1,$ $s=\frac{q^{2}-q}{2}.$
The multiplicative order of $q$ modulo $n$ is denoted by $ord_{n}\left(
q\right) .$ Let $\alpha \in \mathbb{F}_{q^{2}}^{\ast }$ be a primitive $%
r^{th}$ root of unity.

\subsection{EAQMDS codes of length $n=\frac{q^{2}+1}{5},$ where $q\equiv
2\left( mod\text{ }10\right) $}

Note that $n=\frac{q^{2}+1}{5},$ and so $ord_{rn}\left( q^{2}\right) =2.$
This means that each $q^{2}$-cyclotomic coset modulo $rn$ includes one or
two elements. Let $q\equiv 2\left( mod\text{ }10\right) $ and $s=\frac{%
q^{2}-q}{2}.$ It is easy to see that the $q^{2}$-cyclotomic cosets modulo $%
rn $ containing some integers from $1$ to $rn$ are $C_{s-rj}=\left\{
s-rj,s+r\left( j+1\right) \right\} ,$ where $0\leq j\leq \frac{q-2}{2}.$

\begin{lemma}
\label{lm2} Let $q\equiv 2\left( mod\text{ }10\right) .$ If $\mathcal{C}$ is
a $q^{2}$-ary constacyclic code of length $n$ and defining set $Z=\cup
_{j=0}^{\lambda }C_{s-rj},$ where $0\leq \lambda \leq \frac{3\left(
q-2\right) }{10}-1,$ then $\mathcal{C}^{\perp _{H}}\subseteq \mathcal{C}.$
\end{lemma}

\begin{proof}
By Proposition \ref{prop2}, it is sufficient to prove that $Z\cap \left(
-qZ\right) =\emptyset .$ Assume that $Z\cap \left( -qZ\right) \neq \emptyset
.$ Then, there exists two integers $j,k,$ where $0\leq j,k\leq \frac{3\left(
q-2\right) }{10}-1,$ such that $s-rj\equiv -q\left( s-rk\right) \left( mod%
\text{ }rn\right) $ or $s-rj\equiv -q\left( s+rk\right) $ $\left( mod\text{ }%
rn\right) .$

\textbf{Case 1:} Let $s-rj\equiv -q\left( s-rk\right) \left( mod\text{ }%
rn\right) .$ This is equivalent to $s\equiv j+qk\left( mod\text{ }n\right) .$
As $s\equiv r\frac{\left( n-1\right) }{2}\left( mod\text{ }n\right) ,$ we
get $j+qk\equiv r\frac{\left( n-1\right) }{2}\left( mod\text{ }n\right) .$
Since $0\leq j,k\leq \frac{3q-16}{10},$ $0\leq j+qk\leq \frac{3q-16}{10}+q%
\frac{3q-16}{10}<\left( q+1\right) \frac{\left( q-2\right) }{10}<3\left(
q+1\right) \frac{\left( n-1\right) }{2}.$ Then, we have that $j+qk\equiv r%
\frac{\left( n-1\right) }{2}\left( mod\text{ }n\right) $ if and only if $%
j+qk=r\frac{\left( n-1\right) }{2}.$ Since $r\frac{\left( n-1\right) }{2}=r%
\frac{\frac{q^{2}+1}{5}-1}{2}=r\frac{q^{2}-4}{10}=r\frac{q^{2}-2q+2q-4}{10}%
=qr\frac{\left( q-2\right) }{10}+2r\frac{\left( q-2\right) }{10},$ we obtain
$j+qk=qr\frac{\left( q-2\right) }{10}+2r\frac{\left( q-2\right) }{10}.$ By
division algorithm, $j=2r\frac{\left( q-2\right) }{10}.$ This is a
contradiction, because $0\leq j\leq \frac{3\left( q-2\right) }{10}-1.$

\textbf{Case 2:} Let $s-rj\equiv -q\left( s+rk\right) \left( mod\text{ }%
rn\right) .$ This is equivalent to $s\equiv j-qk\left( mod\text{ }n\right) .$
By $s\equiv r\frac{\left( n-1\right) }{2}\left( mod\text{ }n\right) ,$ we
have $j-qk\equiv r\frac{\left( n-1\right) }{2}\left( mod\text{ }n\right) .$
Since $0\leq j,k\leq \frac{3\left( q-2\right) }{10}-1,$ $-3\left( q+1\right)
\frac{\left( n-1\right) }{2}<-q\frac{3q-16}{10}\leq j-qk\leq \frac{3q-16}{10}%
<\left( q+1\right) \frac{\left( n-1\right) }{2}.$ We have that $j-qk\equiv r%
\frac{\left( n-1\right) }{2}\left( mod\text{ }n\right) $ if and only if $%
j-qk=-r\frac{\left( n+1\right) }{2}.$ Since $-r\frac{\left( n+1\right) }{2}%
=-r\frac{\frac{q^{2}+1}{5}+1}{2}=-r\frac{q^{2}+6}{10}=-r\frac{q^{2}-2q+2q+6}{%
10}=-qr\frac{\left( q-2\right) }{10}-r\frac{\left( q+3\right) }{5},$ we
obtain $j-qk=-qr\frac{\left( q-2\right) }{10}-r\frac{\left( q+3\right) }{5}.$
By division algorithm, $k=\frac{\left( q+1\right) \left( q-2\right) }{10}.$
This is a contradiction, because $0\leq k\leq \frac{3\left( q-2\right) }{10}%
-1.$
\end{proof}

\begin{theorem}
\label{th2} Let $q\equiv 2\left( mod\text{ }10\right) .$ If $\mathcal{C}$ is
an $q^{2}$-ary $\alpha $-constacyclic code of length $n$ with defining set $%
Z=$ $\cup _{j=0}^{\frac{3\left( q-2\right) }{10}-1+\lambda }C_{s-rj},$ then
there exists EAQMDS codes with parameters
\[
\llbracket n,n-\frac{6}{5}\left( q-2\right) -4\lambda +4,\frac{3}{5}\left(
q-2\right) +2\lambda +1;4\rrbracket_{q},
\]%
where $1\leq \lambda \leq \frac{q+3}{5}.$
\end{theorem}

\begin{proof}
Since the defining set of $\alpha $-constacyclic code $\mathcal{C}$ of
length $n$ is $Z=$ $\cup _{j=0}^{\frac{3}{10}\left( q-2\right) -1+\lambda
}C_{s-rj},$ and the cardinality of $Z,$ which we denoted by $\left\vert
Z\right\vert ,$ is $\left\vert Z\right\vert =\frac{3}{5}\left( q-2\right)
+2\lambda ,$ then by Proposition \ref{prop1} and \ref{prop3}, $\mathcal{C}$
is an $q^{2}$-ary MDS $\alpha $-constacyclic code with parameters
\[
\left[ n,n-\left( \frac{3}{5}\left( q-2\right) +2\lambda \right) ,\frac{3}{5}%
\left( q-2\right) +2\lambda +1\right] _{q^{2}}.
\]%
Hence, we have the following:%
\begin{eqnarray*}
Z_{1} &=&Z\cap (-qZ)= \\
&&\left( \left( \cup _{j=0}^{\frac{3}{10}\left( q-2\right)
-1}C_{s-rj}\right) \cup \left( \cup _{j=\frac{3}{10}\left( q-2\right) }^{%
\frac{3}{10}\left( q-2\right) -1+\lambda }C_{s-rj}\right) \right) \cap \\
&&\left( -q\left( \cup _{j=0}^{\frac{3}{10}\left( q-2\right)
-1}C_{s-rj}\right) \cup -q\left( \cup _{j=\frac{3}{10}\left( q-2\right) }^{%
\frac{3}{10}\left( q-2\right) -1+\lambda }C_{s-rj}\right) \right) \\
&=&\left( \left( \cup _{j=0}^{\frac{3}{10}\left( q-2\right)
-1}C_{s-rj}\right) \cap -q\left( \cup _{j=0}^{\frac{3}{10}\left( q-2\right)
-1}C_{s-rj}\right) \right) \cup \\
&&\left( \left( \cup _{j=0}^{\frac{3}{10}\left( q-2\right)
-1}C_{s-rj}\right) \cap -q\left( \cup _{j=\frac{3}{10}\left( q-2\right) }^{%
\frac{3}{10}\left( q-2\right) -1+\lambda }C_{s-rj}\right) \right) \\
&&\cup \left( \left( \cup _{j=\frac{3}{10}\left( q-2\right) }^{\frac{3}{10}%
\left( q-2\right) -1+\lambda }C_{s-rj}\right) \cap -q\left( \cup _{j=0}^{%
\frac{3}{10}\left( q-2\right) -1}C_{s-rj}\right) \right) \\
&&\cup \left( \left( \cup _{j=\frac{3}{10}\left( q-2\right) }^{\frac{3}{10}%
\left( q-2\right) -1+\lambda }C_{s-rj}\right) \cap -q\left( \cup _{j=\frac{3%
}{10}\left( q-2\right) }^{\frac{3}{10}\left( q-2\right) -1+\lambda
}C_{s-rj}\right) \right) .
\end{eqnarray*}%
We claim that
\[
Z_{1}=Z\cap (-qZ)=C_{s-r\frac{3\left( q-2\right) }{10}}\cup C_{s-r\frac{%
\left( q-2\right) }{10}}.
\]%
By Lemma \ref{lm2}, we have $\left( \cup _{j=0}^{\frac{3}{10}\left(
q-2\right) -1}C_{s-rj}\right) \cap -q\left( \cup _{j=0}^{\frac{3}{10}\left(
q-2\right) -1}C_{s-rj}\right) =\emptyset .$

We need to show that%
\[
\left( \cup _{j=0}^{\frac{3}{10}\left( q-2\right) -1}C_{s-rj}\right) \cap
-q\left( \cup _{j=\frac{3}{10}\left( q-2\right) }^{\frac{3}{10}\left(
q-2\right) -1+\lambda }C_{s-rj}\right) =C_{s-r\frac{\left( q-2\right) }{10}%
},
\]%
\[
\left( \cup _{j=\frac{3}{10}\left( q-2\right) }^{\frac{3}{10}\left(
q-2\right) -1+\lambda }C_{s-rj}\right) \cap -q\left( \cup _{j=0}^{\frac{3}{10%
}\left( q-2\right) -1}C_{s-rj}\right) =C_{s-r\frac{3\left( q-2\right) }{10}%
},
\]%
\[
\left( \cup _{j=\frac{3}{10}\left( q-2\right) }^{\frac{3}{10}\left(
q-2\right) -1+\lambda }C_{s-rj}\right) \cap -q\left( \cup _{j=\frac{3}{10}%
\left( q-2\right) }^{\frac{3}{10}\left( q-2\right) -1+\lambda
}C_{s-rj}\right) =\emptyset .
\]

We first show that
\[
\left( \cup _{j=\frac{3}{10}\left( q-2\right) }^{\frac{3}{10}\left(
q-2\right) -1+\lambda }C_{s-rj}\right) \cap -q\left( \cup _{j=0}^{\frac{3}{10%
}\left( q-2\right) -1}C_{s-rj}\right) =C_{s-r\frac{3\left( q-2\right) }{10}%
}.
\]

We have the following:%
\begin{eqnarray*}
&&\left( \cup _{j=\frac{3}{10}\left( q-2\right) }^{\frac{3}{10}\left(
q-2\right) -1+\lambda }C_{s-rj}\right) \cap -q\left( \cup _{j=0}^{\frac{3}{10%
}\left( q-2\right) -1}C_{s-rj}\right) \\
&=&\left( C_{s-r\frac{3\left( q-2\right) }{10}}\cup \left( \cup _{j=\frac{3}{%
10}\left( q-2\right) +1}^{\frac{3}{10}\left( q-2\right) -1+\lambda
}C_{s-rj}\right) \right) \cap -q\left( \cup _{j=0}^{\frac{3}{10}\left(
q-2\right) -1}C_{s-rj}\right) \\
&=&\left( C_{s-r\frac{3\left( q-2\right) }{10}}\cap -q\left( \cup _{j=0}^{%
\frac{3}{10}\left( q-2\right) -1}C_{s-rj}\right) \right) \\
&&\cup \left( \left( \cup _{j=\frac{3}{10}\left( q-2\right) +1}^{\frac{3}{10}%
\left( q-2\right) -1+\lambda }C_{s-rj}\right) \cap -q\left( \cup _{j=0}^{%
\frac{3}{10}\left( q-2\right) -1}C_{s-rj}\right) \right) .
\end{eqnarray*}%
We claim that
\[
\left( C_{s-r\frac{3\left( q-2\right) }{10}}\cap -q\left( \cup _{j=0}^{\frac{%
3}{10}\left( q-2\right) -1}C_{s-rj}\right) \right) =C_{s-r\frac{3\left(
q-2\right) }{10}}
\]%
and%
\[
\left( \cup _{j=\frac{3}{10}\left( q-2\right) +1}^{\frac{3}{10}\left(
q-2\right) -1+\lambda }C_{s-rj}\right) \cap -q\left( \cup _{j=0}^{\frac{3}{10%
}\left( q-2\right) -1}C_{s-rj}\right) =\emptyset ,
\]%
where $1\leq \lambda \leq \frac{q+3}{5}.$

Contrary to the claim, assume that
\begin{eqnarray*}
&&\left( \cup _{j=\frac{3}{10}\left( q-2\right) +1}^{\frac{3}{10}\left(
q-2\right) -1+\lambda }C_{s-rj}\right) \cap -q\left( \cup _{j=0}^{\frac{3}{10%
}\left( q-2\right) -1}C_{s-rj}\right) \\
&=&\left( \cup _{j=2}^{\lambda }C_{s-r\left( j+\frac{3}{10}\left( q-2\right)
-1\right) }\right) \cap -q\left( \cup _{j=0}^{\frac{3}{10}\left( q-2\right)
-1}C_{s-rj}\right) \neq \emptyset ,
\end{eqnarray*}%
where $1\leq \lambda \leq \frac{q+3}{5}.$ If $\left( \cup _{j=2}^{\lambda
}C_{s-r\left( j+\frac{3}{10}\left( q-2\right) -1\right) }\right) \cap
-q\left( \cup _{j=0}^{\frac{3}{10}\left( q-2\right) -1}C_{s-rj}\right) \neq
\emptyset ,$ then there exists two integers $u$ and $v,$ where $1\leq u\leq
\frac{q+3}{5},$ and $0\leq v\leq \frac{3\left( q-4\right) }{10}$ such that $%
s-r\left( u+\frac{3}{10}\left( q-2\right) -1\right) \equiv -q\left(
s-rv\right) $ $\left( mod\text{ }rn\right) $ or $s-r\left( u+\frac{3}{10}%
\left( q-2\right) -1\right) \equiv -q\left( s+r\left( v+1\right) \right) $ $%
\left( mod\text{ }rn\right) .$

\textbf{Case 1:} Let $s-r\left( u+\frac{3}{10}\left( q-2\right) -1\right)
\equiv -q\left( s-rv\right) $ $\left( mod\text{ }rn\right) .$ It follows
that $s\equiv \left( u+\frac{3}{10}\left( q-2\right) +qv-1\right) \ \left(
mod\text{ }n\right) .$ Then, $\frac{3\left( q-2\right) }{10}\leq u+\frac{3}{%
10}\left( q-2\right) +qv-1<\frac{2\left( q+3\right) }{10}+\frac{3}{10}\left(
q-2\right) +q\frac{3\left( q-4\right) }{10}=\frac{3q^{2}-7q}{10}<\frac{%
q^{2}-2q}{2}=s-\frac{q}{2},$ which is in contradiction with $s\equiv \frac{%
n-r}{2}$ $\left( mod\text{ }n\right) .$

\textbf{Case 2:} Let $s-r\left( u+\frac{3}{10}\left( q-2\right) -1\right)
\equiv -q\left( s+r\left( v+1\right) \right) $ $\left( mod\text{ }rn\right)
. $ Then, $s\equiv \left( u+\frac{3}{10}q-2-q\left( v+1\right) -1\right) \
\left( mod\text{ }n\right) .$ It follows that $-q\frac{3\left( q-4\right) }{%
10}+\frac{3\left( q-2\right) }{10}-q\leq u+\frac{3}{10}\left( q-2\right)
-q\left( v+1\right) -1<\frac{2\left( q+3\right) }{10}+\frac{3}{10}\left(
q-2\right) -q=\frac{-q}{2},$ which is a contradiction, since $s\equiv \frac{%
n-r}{2}$ $\left( mod\text{ }n\right) .$

The above discussions show that
\[
\left( \cup _{j=\frac{3}{10}\left( q-2\right) }^{\frac{3}{10}\left(
q-2\right) -1+\lambda }C_{s-rj}\right) \cap -q\left( \cup _{j=0}^{\frac{3}{10%
}\left( q-2\right) -1}C_{s-rj}\right) =C_{s-r\frac{3\left( q-2\right) }{10}}
\]%
Since $-qs\equiv s-r\frac{\left( 3q+1\right) \left( q-2\right) }{10}$ $%
\left( mod\text{ }rn\right) ,$ we have $-q\left( s-r\frac{3\left( q-2\right)
}{10}\right) \equiv s-r\frac{\left( q-2\right) }{10}$ $\left( mod\text{ }%
rn\right) .$ This fact says that

\[
-q\left( \cup _{j=\frac{3}{10}\left( q-2\right) }^{\frac{3}{10}\left(
q-2\right) -1+\lambda }C_{s-rj}\right) \cap -q\left( \cup _{j=0}^{\frac{3}{10%
}\left( q-2\right) -1}C_{s-rj}\right) =-qC_{s-r\frac{3\left( q-2\right) }{10}%
}=C_{s-r\frac{\left( q-2\right) }{10}},
\]%
and it follows that
\[
\left( \cup _{j=0}^{\frac{3}{10}\left( q-2\right) -1}C_{s-rj}\right)
\cap-q\left( \cup _{j=\frac{3}{10}\left( q-2\right) }^{\frac{3}{10}%
\left(q-2\right) -1+\lambda }C_{s-rj}\right) =C_{s-r\frac{\left( q-2\right)
}{10}}.
\]

Contrary to the claim, suppose that%
\begin{eqnarray*}
&&\left( \cup _{j=\frac{3}{10}\left( q-2\right) }^{\frac{3}{10}\left(
q-2\right) -1+\lambda }C_{s-rj}\right) \cap -q\left( \cup _{j=\frac{3}{10}%
\left( q-2\right) }^{\frac{3}{10}\left( q-2\right) -1+\lambda
}C_{s-rj}\right) \\
&=&\left( \cup _{j=1}^{\lambda }C_{s-r\left( j+\frac{3}{10}\left( q-2\right)
-1\right) }\right) \cap -q\left( \cup _{j=1}^{\lambda }C_{s-r\left( j+\frac{3%
}{10}\left( q-2\right) -1\right) }\right) \neq \emptyset ,
\end{eqnarray*}%
where $1\leq \lambda \leq \frac{q+3}{5}.$ Then, there exists two integers $u$
and $v,$ where $1\leq u,v\leq \frac{q+3}{5},$ such that $s-r\left( u+\frac{3%
}{10}\left( q-2\right) -1\right) \equiv -q\left( s-\left( v+\frac{3}{10}%
\left( q-2\right) -1\right) \right) $ $\left( mod\text{ }rn\right) $ or $%
s-r\left( u+\frac{3}{10}\left( q-2\right) -1\right) \equiv -q\left( s+\left(
v+\frac{3}{10}\left( q-2\right) \right) \right) $ $\left( mod\text{ }%
rn\right) .$

\textbf{Case 1:} Let $s-r\left( u+\frac{3}{10}\left( q-2\right) -1\right)
\equiv -q\left( s-r\left( v+\frac{3}{10}\left( q-2\right) -1\right) \right) $
$\left( mod\text{ }rn\right) .$ It follows that $s\equiv \left( u+qv+\frac{3%
}{10}\left( q+1\right) \left( 3q-16\right) \right) \ \left( mod\text{ }%
n\right) .$ Then, we get $\frac{3\left( q-2\right) \left( q+1\right) }{10}%
\leq u+\frac{\left( q+1\right) \left( 3q-16\right) }{10}+qv\leq \frac{\left(
q+1\right) \left( 5q-10\right) }{10}=\frac{q^{2}-q-2}{2}=s-1.$ This
contradicts with $s\equiv \frac{n-r}{2}$ $\left( mod\text{ }n\right) .$

\textbf{Case 2:} Let $s-r\left( u+\frac{3}{10}\left( q-2\right) -1\right)
\equiv -q\left( s+r\left( v+\frac{3}{10}\left( q-2\right) \right) \right) .$
It follows that $s\equiv \left( u-qv+\left( 1-q\right) \frac{3\left(
q-2\right) }{10}-1\right) \ \left( mod\text{ }n\right) .$ Then, we have $%
\frac{-3q^{2}+3q-6}{10}\leq u-qv+\left( 1-q\right) \frac{3\left( q-2\right)
}{10}-1=\frac{-3q^{2}-9q}{10},$ which is in contradiction with $s\equiv
\frac{n-r}{2}$ $\left( mod\text{ }n\right) .$

This means that%
\[
\left( \cup _{j=\frac{3}{10}\left( q-2\right) }^{\frac{3}{10}%
\left(q-2\right) -1+\lambda }C_{s-rj}\right) \cap -q\left( \cup _{j=\frac{3}{%
10}\left( q-2\right) }^{\frac{3}{10}\left(q-2\right)
-1+\lambda}C_{s-rj}\right) =\emptyset.
\]

From Lemma \ref{lm1},\ we have $c=4,$ and by Theorem \ref{th1}, there exists
EAQMDS codes with parameters
\[
\llbracket  n,n-\frac{6}{5}\left( q-2\right) -4\lambda +4,\frac{3}{5}\left(
q-2\right) +2\lambda +1;4 \rrbracket _{q},
\]
where $1\leq \lambda \leq \frac{q+3}{5}.$
\end{proof}

\begin{example}
We present some parameters of EAQMDS codes obtained from Theorem \ref{th2}
in Table \ref{table1}.
\end{example}

\begin{table}[tbp]
\caption{Some EAQMDS codes obtained by Theorem \protect\ref{th2}.}
\label{table1}\centering
\begin{tabular}{ccc}
\hline
& $\lambda $ & $\llbracket n,n-\frac{6}{5}\left( q-2\right) -4\lambda +4,%
\frac{3}{5}\left( q-2\right) +2\lambda +1;4 \rrbracket _{q}$ \\ \hline
& $1$ & $\llbracket 205{,169,21;4} \rrbracket _{32}$ \\
& $2$ & $\llbracket 205{,165,23;4} \rrbracket _{32}$ \\
& $3$ & $\llbracket 205{,161,25;4} \rrbracket _{32}$ \\
& $4$ & $\llbracket 205{,157,27;4} \rrbracket _{32}$ \\
& $5$ & $\llbracket 205{,153,29;4} \rrbracket _{32}$ \\
& $6$ & $\llbracket 205{,149,31;4} \rrbracket _{32}$ \\
& $7$ & $\llbracket 205{,145,33;4} \rrbracket _{32}$ \\ \hline
\end{tabular}%
\end{table}

\subsection{EAQMDS codes of length $n=\frac{q^{2}+1}{5},$ where $q\equiv
8\left( mod\text{ }10\right) $}

Let $q\equiv 8\left( mod\text{ }10\right) $ and $s=\frac{q^{2}-q}{2}.$ Then,
the $q^{2}$-cyclotomic cosets modulo $rn$ containing some integers from $1$
to $rn$ are $C_{s-rj}=\left\{ s-rj,s+r\left( j+1\right) \right\} ,$ where $%
0\leq j\leq \frac{q-2}{2}.$

\begin{lemma}
\label{lm3} Let $q\equiv 8\left( mod\text{ }10\right) .$ If $\mathcal{C}$ is
a $q^{2}$-ary constacyclic code of length $n$ and its defining set is $%
Z=\cup _{j=0}^{\lambda }C_{s-rj},$ where $0\leq \lambda \leq \frac{3q-14}{10}%
,$ then $\mathcal{C}^{\perp _{H}}\subseteq \mathcal{C}.$
\end{lemma}

\begin{proof}
By Proposition \ref{prop2}, it is sufficient to prove that $Z\cap \left(
-qZ\right) =\emptyset .$ Assume that $Z\cap \left( -qZ\right) \neq \emptyset
.$ Then, there exists two integers $j,$ $k,0\leq j,k\leq \frac{3q-14}{10},$
such that $s-rj\equiv -q\left( s-rk\right) $ $\left( mod\text{ }rn\right) $
or $s-rj\equiv -q\left( s+rk\right) \left( mod\text{ }rn\right) .$

\textbf{Case 1:} $s-rj\equiv -q\left( s-rk\right) $ $\left( mod\text{ }%
rn\right) .$ This is equivalent to $s\equiv j+qk$ $\left( mod\text{ }%
n\right) .$ As $s\equiv r\frac{\left( n-1\right) }{2}$ $\left( mod\text{ }%
n\right) ,$ we get $j+qk\equiv r\frac{\left( n-1\right) }{2}$ $\left( mod%
\text{ }n\right) .$ Since $0\leq j,k\leq \frac{3q-14}{10},$ $0\leq j+qk\leq
\frac{3q-14}{10}+q\frac{3q-14}{10}\leq \left( q+1\right) \frac{\left(
3q-14\right) }{10}<3r\frac{\left( n-1\right) }{2}.$ Then, we have that $%
j+qk\equiv r\frac{\left( n-1\right) }{2}$ $\left( mod\text{ }n\right) $ if
and only if $j+qk=r\frac{\left( n-1\right) }{2}+nt,$ for some integer $t.$
Then, we obtain $r\frac{\left( n-1\right) }{2}=r\frac{\left( \frac{q^{2}+1}{5%
}-1\right) }{2}=r\frac{\left( q^{2}-4\right) }{10}.$ This is a
contradiction, because $0\leq j\leq \frac{3q-14}{10}<\frac{\left(
q^{2}-4\right) }{10}.$

\textbf{Case 2:} $s-rj\equiv -q\left( s+rk\right) $ $\left( mod\text{ }%
rn\right) .$ This is equivalent to $s\equiv j-qk$ $\left( mod\text{ }%
n\right) .$ Since $s\equiv r\frac{\left( n-1\right) }{2}\left( mod\text{ }%
n\right) ,$ we have $j-qk\equiv r\frac{\left( n-1\right) }{2}$ $\left( mod%
\text{ }n\right) ,$ where $0\leq j,k\leq \frac{3q-14}{10}.$ Then, $-q\frac{%
3q-14}{10}\leq j-qk\leq \frac{3q-14}{10}.$ We have that $j-qk\equiv r\frac{%
\left( n-1\right) }{2}$ $\left( mod\text{ }n\right) $ if and only if $j-qk=-r%
\frac{\left( n+1\right) }{2}+nt,$ for some integer $t.$ Then, $-r\frac{%
\left( n+1\right) }{2}=-r\frac{\frac{q^{2}+1}{5}+1}{2}=-r\frac{q^{2}+6}{10}.$
This is a contradiction, because $-q\frac{3q-14}{10}\leq j-qk\leq \frac{3q-14%
}{10}.$
\end{proof}

\begin{theorem}
\label{th3} Let $q\equiv 8\left( mod\text{ }10\right) .$ If $\mathcal{C}$ is
an $q^{2}$-ary $\alpha $-constacyclic code of length $n$ with defining set $%
Z=\cup _{j=0}^{\frac{3q-14}{10}+\lambda }C_{s-rj},$ then there exists EAQMDS
codes with parameters
\[
\llbracket  n,n-\frac{2}{5}\left( 3q-14\right) -4\lambda ,\frac{%
\left(3q-14\right) }{5}+2\lambda +3;4\rrbracket_{q},
\]%
where $1\leq \lambda \leq \frac{q+2}{5}.$
\end{theorem}

\begin{proof}
Because, the defining set of $\alpha $-constacyclic code $\mathcal{C}$ of
length $n$ is $Z=$ $\cup _{j=0}^{\frac{3q-14}{10}+\lambda }C_{s-rj},$ then
the cardinality of $Z$ is $\left\vert Z\right\vert =$ $\frac{3q-14}{5}%
+2\lambda +2.$ From Proposition \ref{prop1} and \ref{prop3}, $\mathcal{C}$
is an MDS $\alpha $-constacyclic code with parameters
\[
\left[ n,n-\left( \frac{3q-14}{5}+2\lambda +2\right) ,\frac{\left(
3q-14\right) }{5}+2\lambda +3\right] _{q^{2}}.
\]%
Thus, we have the following:%
\begin{eqnarray*}
Z_{1} &=&Z\cap (-qZ)= \\
&&\left( \left( \cup _{j=0}^{\frac{3q-14}{10}}C_{s-rj}\right) \cup \left(
\cup _{j=\frac{3q-4}{10}}^{\frac{3q-14}{10}+\lambda }C_{s-rj}\right) \right)
\cap \\
&&\left( -q\left( \cup _{j=0}^{\frac{3q-14}{10}}C_{s-rj}\right) \cup
-q\left( \cup _{j=\frac{3q-4}{10}}^{\frac{3q-14}{10}+\lambda
}C_{s-rj}\right) \right) \\
&=&\left( \left( \cup _{j=0}^{\frac{3q-14}{10}}C_{s-rj}\right) \cap -q\left(
\cup _{j=0}^{\frac{3q-14}{10}}C_{s-rj}\right) \right) \cup \\
&&\left( \left( \cup _{j=0}^{\frac{3q-14}{10}}C_{s-rj}\right) \cap -q\left(
\cup _{j=\frac{3q-4}{10}}^{\frac{3q-14}{10}+\lambda }C_{s-rj}\right) \right)
\\
&&\cup \left( \left( \cup _{j=\frac{3q-4}{10}}^{\frac{3q-14}{10}+\lambda
}C_{s-rj}\right) \cap -q\left( \cup _{j=0}^{\frac{3q-14}{10}}C_{s-rj}\right)
\right) \\
&&\cup \left( \left( \cup _{j=\frac{3q-4}{10}}^{\frac{3q-14}{10}+\lambda
}C_{s-rj}\right) \cap -q\left( \cup _{j=\frac{3q-4}{10}}^{\frac{3q-14}{10}%
+\lambda }C_{s-rj}\right) \right) .
\end{eqnarray*}%
We claim that
\[
Z_{1}=Z\cap (-qZ)=C_{s-r\frac{\left( q-8\right) }{10}}\cup C_{s-r\frac{%
\left( 3q-4\right) }{10}}.
\]%
By Lemma \ref{lm2}, we have $\left( \cup _{j=0}^{\frac{3q-14}{10}%
}C_{s-rj}\right) \cap -q\left( \cup _{j=0}^{\frac{3q-14}{10}}C_{s-rj}\right)
=\emptyset .$

It is sufficient to show that%
\[
\left( \cup _{j=0}^{\frac{3q-14}{10}}C_{s-rj}\right) \cap -q\left( \cup
_{j=0}^{\frac{3q-14}{10}}C_{s-rj}\right) =C_{s-r\frac{\left( q-8\right) }{10}%
},
\]%
\[
\left( \cup _{j=\frac{3q-4}{10}}^{\frac{3q-14}{10}+\lambda }C_{s-rj}\right)
\cap -q\left( \cup _{j=0}^{\frac{3q-14}{10}}C_{s-rj}\right) =C_{s-r\frac{%
\left( 3q-4\right) }{10}},
\]%
\[
\left( \cup _{j=\frac{3q-4}{10}}^{\frac{3q-14}{10}+\lambda }C_{s-rj}\right)
\cap -q\left( \cup _{j=\frac{3q-4}{10}}^{\frac{3q-14}{10}+\lambda
}C_{s-rj}\right) =\emptyset .
\]

We first show that
\[
\left( \cup _{j=\frac{3q-4}{10}}^{\frac{3q-14}{10}+\lambda
}C_{s-rj}\right)\cap -q\left( \cup _{j=0}^{\frac{3q-14}{10}}C_{s-rj}\right)
=C_{s-r\frac{\left( 3q-4\right) }{10}}.
\]

The following is immediate:%
\begin{eqnarray*}
&&\left( \cup _{j=\frac{3q-4}{10}}^{\frac{3q-14}{10}+\lambda
}C_{s-rj}\right) \cap -q\left( \cup _{j=0}^{\frac{3q-14}{10}}C_{s-rj}\right)
\\
&=&\left( C_{s-r\frac{\left( 3q-4\right) }{10}}\cup \left( \cup _{j=\frac{%
3q+6}{10}}^{\frac{3q-14}{10}+\lambda }C_{s-rj}\right) \right) \cap -q\left(
\cup _{j=0}^{\frac{3q-14}{10}}C_{s-rj}\right) \\
&=&\left( C_{s-r\frac{\left( 3q-4\right) }{10}}\cap -q\left( \cup _{j=0}^{%
\frac{3q-14}{10}}C_{s-rj}\right) \right) \\
&&\cup \left( \left( \cup _{j=\frac{3q+6}{10}}^{\frac{3q-14}{10}+\lambda
}C_{s-rj}\right) \cap -q\left( \cup _{j=0}^{\frac{3q-14}{10}}C_{s-rj}\right)
\right) .
\end{eqnarray*}%
We claim that
\[
C_{s-r\frac{\left( 3q-4\right) }{10}}\cap -q\left( \cup _{j=0}^{\frac{3q-14}{%
10}}C_{s-rj}\right) =C_{s-r\frac{\left( 3q-4\right) }{10}},
\]%
and
\[
\left( \cup _{j=\frac{3q+6}{10}}^{\frac{3q-14}{10}+\lambda }C_{s-rj}\right)
\cap -q\left( \cup _{j=0}^{\frac{3q-14}{10}}C_{s-rj}\right) =\emptyset ,
\]%
where $1\leq \lambda \leq \frac{q+2}{5}.$

Contrary to the our claim, assume that
\begin{eqnarray*}
&&\left( \cup _{j=\frac{3q+6}{10}}^{\frac{3q-14}{10}+\lambda
}C_{s-rj}\right) \cap -q\left( \cup _{j=0}^{\frac{3q-14}{10}}C_{s-rj}\right)
\\
&=&\left( \cup _{j=2}^{\lambda }C_{s-r\left( j+\frac{3q-14}{10}\right)
}\right) \cap -q\left( \cup _{j=0}^{\frac{3q-14}{10}}C_{s-rj}\right) \neq
\emptyset ,
\end{eqnarray*}%
where $1\leq \lambda \leq \frac{q+2}{5}.$ If $\left( \cup _{j=2}^{\lambda
}C_{s-r\left( j+\frac{3q-14}{10}\right) }\right) \cap -q\left( \cup _{j=0}^{%
\frac{3q-14}{10}}C_{s-rj}\right) \neq \emptyset ,$ then there exists two
integers $a$ and $b,$ where $1\leq a\leq \frac{q+2}{5},$ and $0\leq b\leq
\frac{\left( 3q-14\right) }{10}$ such that $s-r\left( a+\frac{3q-14}{10}%
\right) \equiv -q\left( s-rb\right) $ $\left( mod\text{ }rn\right) $ or $%
s-r\left( a+\frac{3q-14}{10}\right) \equiv -q\left( s+r\left( b+1\right)
\right) $ $\left( mod\text{ }rn\right) .$

\textbf{Case 1:} Let $s-r\left( a+\frac{3q-14}{10}\right) \equiv -q\left(
s-rb\right) $ $\left( mod\text{ }rn\right) .$ It follows that $s\equiv
\left( a+\frac{3q-14}{10}+qb\right) \ \left( mod\text{ }n\right) .$ Then, $%
\frac{\left( 3q-4\right) }{10}\leq a+\frac{3q-14}{10}+qb\leq \frac{3q^{2}-9q%
}{10}-1<\frac{q^{2}-2q}{2}=s-\frac{q}{2},$ which is in contradiction with $%
s\equiv \frac{n-r}{2}$ $\left( mod\text{ }n\right) .$

\textbf{Case 2:} Let $s-r\left( a+\frac{3q-14}{10}\right) \equiv -q\left(
s+r\left( b+1\right) \right) $ $\left( mod\text{ }rn\right) .$ It follows
that $s\equiv \left( a+\frac{3q-14}{10}-q\left( b+1\right) \right) \ \left(
mod\text{ }n\right) .$ Then, $\frac{-3q^{2}+3q}{10}-1\leq a+\frac{3q-14}{10}%
-q\left( b+1\right) \leq -\frac{q}{2}-1,$ which is a contradiction, since $%
s\equiv \frac{n-r}{2}$ $\left( mod\text{ }n\right) .$

The above arguments show that
\[
\left( \cup _{j=\frac{3q-4}{10}}^{\frac{3q-14}{10}+\lambda }C_{s-rj}\right)
\cap -q\left( \cup _{j=0}^{\frac{3q-14}{10}}C_{s-rj}\right) =C_{s-r\frac{%
\left( 3q-4\right) }{10}}.
\]%
We have $-q\left( s-r\frac{3q-4}{10}\right) \equiv s-r\frac{q-8}{10}$ $%
\left( mod\text{ }rn\right) .$ This means that

\[
-q\left( \cup _{j=\frac{3q-4}{10}}^{\frac{3q-14}{10}+\lambda
}C_{s-rj}\right) \cap -q\left( \cup _{j=0}^{\frac{3q-14}{10}}C_{s-rj}\right)
=-qC_{s-r\frac{3q-4}{10}}=C_{s-r\frac{q-8}{10}},
\]%
and it follows that
\[
\left( \cup _{j=0}^{\frac{3q-14}{10}}C_{s-rj}\right) \cap -q\left( \cup
_{j=0}^{\frac{3q-14}{10}}C_{s-rj}\right) =C_{s-r\frac{\left( q-8\right) }{10}%
}.
\]

For the remaining part of the proof, suppose that%
\begin{eqnarray*}
&&\left( \cup _{j=\frac{3q-4}{10}}^{\frac{3q-14}{10}+\lambda
}C_{s-rj}\right) \cap -q\left( \cup _{j=\frac{3q-4}{10}}^{\frac{3q-14}{10}%
+\lambda }C_{s-rj}\right) \\
&=&\left( \cup _{j=1}^{\lambda }C_{s-r\left( j+\frac{3q-14}{10}\right)
}\right) \cap -q\left( \cup _{j=1}^{\lambda }C_{s-r\left( j+\frac{3q-14}{10}%
\right) }\right) \neq \emptyset ,
\end{eqnarray*}
where $1\leq \lambda \leq \frac{q+2}{5}.$ Then, there exists two integers $a$
and $b,$ where $2\leq a,b\leq \frac{q+2}{5},$ such that $s-r\left( a+\frac{%
3q-14}{10}\right) \equiv -q\left( s-\left( b+\frac{3q-14}{10}\right) \right)
$ $\left( mod\text{ }rn\right) $ or $s-r\left( a+\frac{3q-14}{10}\right)
\equiv -q\left( s+r\left( b+1+\frac{3q-14}{10}\right) \right) $ $\left( mod%
\text{ }rn\right).$

\textbf{Case 1:} Let $s-r\left( a+\frac{3q-14}{10}\right) \equiv -q\left(
s-\left( b+\frac{3q-14}{10}\right) \right) $ $\left( mod\text{ }rn\right) .$
It follows that $s\equiv \left( a+qb+\left( q+1\right) \frac{3q-14}{10}%
\right) \ \left( mod\text{ }n\right) .$ It is immediate that $\frac{\left(
q+1\right) \left( 3q-4\right) }{10}\leq a+qb+\left( q+1\right) \frac{3q-14}{%
10}\leq \frac{\left( q+1\right) \left( 5q-10\right) }{10}=\frac{q^{2}-q-2}{2}%
=s-1.$ This contradicts with $s\equiv \frac{n-r}{2}$ $\left( mod\text{ }%
n\right) .$

\textbf{Case 2:} Let $s-r\left( a+\frac{3q-14}{10}\right) \equiv -q\left(
s+r\left( b+1+\frac{3q-14}{10}\right) \right) $ $\left( mod\text{ }rn\right)
.$ It follows that $s\equiv \left( a-q\left( b+1\right) +\left( 1-q\right)
\frac{3q-14}{10}\right) \ \left( mod\text{ }n\right) .$ We have $\frac{%
-5q^{2}+3q-4}{10}$ $\leq a-q\left( b+1\right) +\left( 1-q\right) \frac{3q-14%
}{10}\leq \frac{-3q^{2}-q}{10}-1<s,$ which is in contradiction with $s\equiv
\frac{n-r}{2}$ $\left( mod\text{ }n\right) .$ This means that%
\[
\left( \cup _{j=\frac{3q-4}{10}}^{\frac{3q-14}{10}+\lambda }C_{s-rj}\right)
\cap -q\left( \cup _{j=\frac{3q-4}{10}}^{\frac{3q-14}{10}+\lambda
}C_{s-rj}\right) =\emptyset .
\]%
From Lemma \ref{lm1},\ we have $c=4,$ and by Theorem \ref{th1}, there exists
EAQMDS codes with parameters
\[
\llbracket n,n-\frac{2}{5}\left( 3q-14\right) -4\lambda ,\frac{\left(
3q-14\right) }{5}+2\lambda +3;4\rrbracket_{q},
\]%
where $1\leq \lambda \leq \frac{q+2}{5}.$
\end{proof}

\begin{example}
We present some parameters of EAQMDS codes obtained from Theorem \ref{th3}
in Table \ref{table2}.
\end{example}

\begin{table}[tbp]
\caption{Some EAQMDS codes obtained by Theorem \protect\ref{th3}.}
\label{table2}\centering
\begin{tabular}{ccc}
\hline
$q$ & $\lambda $ & $\llbracket n,n-\frac{2}{5}\left( 3q-14\right) -4\lambda ,%
\frac{\left( 3q-14\right) }{5}+2\lambda +3;4 \rrbracket _{q}$ \\ \hline
$8$ & $1$ & $\llbracket 13,5,7{;4} \rrbracket _{8}$ \\
$8$ & $2$ & $\llbracket 13,1,9{;4} \rrbracket _{8}$ \\
$128$ & $1$ & $\llbracket 3277,3125,79{;4} \rrbracket _{128}$ \\
$128$ & $2$ & $\llbracket 3277,3121,81{;4} \rrbracket _{128}$ \\
$128$ & $3$ & $\llbracket 3277,3117,83{;4} \rrbracket _{128}$ \\
$128$ & $4$ & $\llbracket 3277,3113,85{;4} \rrbracket _{128}$ \\
$128$ & $5$ & $\llbracket 3277,3109,87{;4} \rrbracket _{128}$ \\
$128$ & $6$ & $\llbracket 3277,3105,89{;4} \rrbracket _{128}$ \\
$\vdots $ & $\vdots $ & $\vdots $ \\
$128$ & $26$ & $\llbracket 3277,3025,129{;4} \rrbracket _{128}$ \\ \hline
\end{tabular}%
\end{table}

\subsection{EAQMDS codes of length $n=\frac{q^{2}+1}{13},$ where $q\equiv
5\left( mod\text{ }13\right) $}

Note that $n=\frac{q^{2}+1}{13},$ and so $ord_{rn}\left( q^{2}\right) =2.$
This means that each $q^{2}$-cyclotomic coset modulo $rn$ includes one or
two elements. Let $q\equiv 5\left( mod\text{ }13\right) $ and $s=\frac{%
q^{2}-q}{2}.$ It is easy to see that the $q^{2}$-cyclotomic cosets modulo $%
rn $ containing some integers from $1$ to $rn$ are $C_{s-rj}=\left\{
s-rj,s+r\left( j+1\right) \right\} ,$ where $0\leq j\leq \frac{q-2}{2}.$

\begin{lemma}
\label{lm4} Let $q\equiv 5\left( mod\text{ }13\right) .$ If $\mathcal{C}$ is
a $q^{2}$-ary constacyclic code of length $n$ and defining set $Z=\cup
_{j=0}^{\lambda }C_{s-rj},$ where $0\leq \lambda \leq \frac{3\left(
q-2\right) }{10}-1,$ then $\mathcal{C}^{\perp _{H}}\subseteq \mathcal{C}.$
\end{lemma}

\begin{proof}
The proof is analogous to the proof of the Lemma \ref{lm2}.
\end{proof}

\begin{theorem}
\label{th4} Let $q\equiv 5\left( mod\text{ }13\right) .$ If $\mathcal{C}$ is
an $q^{2}$-ary $\alpha $-constacyclic code of length $n$ with defining set $%
Z=\cup _{j=0}^{\frac{3\left( q-2\right) }{10}-1+\lambda }C_{s-rj},$ then
there exists EAQMDS codes with parameters
\[
\llbracket n,n-\frac{6}{5}\left( q-2\right) -4\lambda +4,\frac{3}{5}\left(
q-2\right) +2\lambda +1;4\rrbracket_{q},
\]%
where $1\leq \lambda \leq \frac{q+3}{5}.$
\end{theorem}

\begin{proof}
The proof is analogous to the proof of Theorem \ref{th2}.
\end{proof}

\begin{example}
We present some parameters of EAQMDS codes obtained from Theorem \ref{th4}
in Table \ref{table3}.
\end{example}

\begin{table}[tbp]
\caption{Some EAQMDS codes obtained by Theorem \protect\ref{th4}.}
\label{table3}\centering
\begin{tabular}{cc}
\hline
$\lambda $ & $\llbracket n,n-\frac{6}{5}\left( q-2\right) -4\lambda +4,\frac{%
3}{5}\left( q-2\right) +2\lambda +1;4\rrbracket_{q}$ \\ \hline
$1$ & $\llbracket20165,19553,309{;4}\rrbracket_{512}$ \\
$2$ & $\llbracket20165,19549,311{;4}\rrbracket_{512}$ \\
$\vdots $ & $\vdots $ \\
$103$ & $\llbracket20165,19145,513{;4}\rrbracket_{512}$ \\ \hline
\end{tabular}%
\end{table}

\subsection{EAQMDS codes of length $n=\frac{q^{2}+1}{17},$ where $q\equiv
13\left( mod\text{ }17\right) $}

Note that $n=\frac{q^{2}+1}{17},$ and so $ord_{rn}\left( q^{2}\right) =2.$
This means that each $q^{2}$-cyclotomic coset modulo $rn$ includes one or
two elements. Let $q\equiv 13\left( mod\text{ }17\right) $ and $s=\frac{%
q^{2}-q}{2}.$ It is easy to see that the $q^{2}$-cyclotomic cosets modulo $%
rn $ containing some integers from $1$ to $rn$ are $C_{s-rj}=\left\{
s-rj,s+r\left( j+1\right) \right\} ,$ where $0\leq j\leq \frac{q-2}{2}.$

\begin{lemma}
\label{lm5} Let $q\equiv 13\left( mod\text{ }17\right) .$ If $\mathcal{C}$
is a $q^{2}$-ary constacyclic code of length $n$ and defining set $Z=\cup
_{j=0}^{\lambda }C_{s-rj},$ where $0\leq \lambda \leq \frac{3\left(
q-4\right) }{10}+2,$ then $\mathcal{C}^{\perp _{H}}\subseteq \mathcal{C}.$
\end{lemma}

\begin{proof}
The proof is analogous to the proof of the Lemma \ref{lm2}.
\end{proof}

\begin{theorem}
\label{th5} Let $q\equiv 13\left( mod\text{ }17\right) .$ If $\mathcal{C}$
is an $q^{2}$-ary $\alpha $-constacyclic code of length $n$ with defining
set $Z=\cup _{j=0}^{\frac{3\left( q-2\right) }{10}+2+\lambda }C_{s-rj},$
then there exists EAQMDS codes with parameters
\[
\llbracket n,n-\frac{6}{5}\left( q-4\right) -4\lambda -8,\frac{3}{5}\left(
q-4\right) +2\lambda +4;4\rrbracket_{q},
\]%
where $1\leq \lambda \leq \frac{q+4}{17}.$
\end{theorem}

\begin{proof}
The proof is analogous to the proof of the Theorem \ref{th2}.
\end{proof}

\begin{example}
We present some parameters of EAQMDS codes obtained from Theorem \ref{th5}
in Table \ref{table4}.
\end{example}

\begin{table}[tbp]
\caption{Some EAQMDS codes obtained by Theorem \protect\ref{th5}.}
\label{table4}\centering
\begin{tabular}{cc}
\hline
$\lambda $ & $\llbracket n,n-\frac{6}{5}\left( q-4\right) -4\lambda -8,\frac{%
3}{5}\left( q-4\right) +2\lambda +4;4\rrbracket_{q}$ \\ \hline
$1$ & $\llbracket241,157,42{;4}\rrbracket_{64}$ \\
$2$ & $\llbracket241,153,44{;4}\rrbracket_{64}$ \\
$3$ & $\llbracket241,149,46{;4}\rrbracket_{64}$ \\
$4$ & $\llbracket241,145,48{;4}\rrbracket_{64}$ \\ \hline
\end{tabular}%
\end{table}

\section{Conclusion}

In this work, via a decomposition of the defining set of constacyclic codes
we have constructed eight new families of EAQMDS codes. In addition to the
parameters of EAQMDS and EAQC codes given in \cite{Lu1} and \cite{Guenda},
we remark that the parameters of EAQMDS and EAQC codes listed below haven't
covered ones given in this paper.

\begin{enumerate}
\item $\llbracket  q^{2}+1,q^{2}-2d+4,d;1\rrbracket  _{q},$ where $q$ is a
prime power, $2\leq d\leq 2q$ is an even integer (\cite{Fan}).

\item $\llbracket  \frac{q^{2}-1}{2},\frac{q^{2}-1}{2}-2d+4,d;2\rrbracket  %
_{q},$ where $q$ is an odd prime power, $\frac{q+5}{2}\leq d\leq \frac{3q-1}{%
2}$ (\cite{Fan}).

\item $\llbracket  n,n-2\delta -1,2\delta +2;2\delta +1\rrbracket  _{q},$
where $q$ is an odd prime power, $n=q^{2}+1,$ $s=\frac{n}{2},r|q-1,$ $r\nmid
q+1,$and $0\leq \delta \leq \frac{(r-1)(s-1)}{r}$ (\cite{Qian2}).

\item $\llbracket  n,n-2\delta -2,2\delta +3;2\delta +2\rrbracket  _{q},$%
where $q=2^{m},n=q^{2}+1,r|q-1,r\nmid q+1,u=\frac{n-r}{2}$ and $0\leq \delta
\leq \frac{u-1}{r}$ (\cite{Qian2}).

\item $\llbracket  n,n-2\delta -1,2\delta +2;2\delta +1\rrbracket  _{q},$
where $q=2^{m},n=q^{2}+1,r|q-1,r\nmid q+1$ and $0\leq \delta \leq \frac{%
(r-1)(n-2)}{2r}$ (\cite{Qian2}).

\item $\llbracket  q^{2}+1,q^{2}+5-2q-4t,q+2t+1;4\rrbracket  _{q},$ where $%
2\leq t\leq \frac{q-1}{2},$ $q$ is an odd prime power with $q\geq 5$ and $%
q\equiv 1$ $\left( mod\text{ }4\right) $ (\cite{ChenJ}).

\item $\llbracket  \frac{q^{2}+1}{2},\frac{q^{2}+1}{2}-2q-4t+5,q+2t+1;5%
\rrbracket  _{q},$ where $2\leq t\leq \frac{q-1}{2},$ $q$ is an odd prime
power with $q>7$ (\cite{ChenJ}).

\item $\llbracket  \lambda(q+1),\lambda(q+1)-2\lambda-2t-q+5,\frac{q+1}{2}%
+t+\lambda ;4\rrbracket  _{q},$ where $q$ is an odd prime power with $q\geq
7,$ $\lambda$ is an odd divisor of $q-1$ with $\lambda \geq 3$ and $\frac{q+3%
}{2}\leq t\leq \frac{q-1}{2}+\lambda $ (\cite{ChenJ}).

\item $\llbracket  2\lambda(q+1),2\lambda(q+1)-4\lambda-2t-q+5,\frac{q+1}{2}%
+t+2\lambda ;4\rrbracket  _{q},$ where $q$ is an odd prime power with $q\geq
13,$ $q\equiv 1$ $\left( mod\text{ }4\right) ,$ $\lambda$ is an odd divisor
of $q-1 $ with $\lambda\geq 3$ and $\frac{q+3}{2}\leq t\leq \frac{q-1}{2}%
+2\lambda$ (\cite{ChenJ}).

\item $\llbracket  \frac{q^{2}+1}{2},\frac{q^{2}+1}{2}-5,d\geq 3;5%
\rrbracket
_{q},$ where $q$ is an odd prime power with $q>3$ (\cite{ChenJ}).

\item $\llbracket  q^{2}+1,q^{2}-3,d\geq 3;4\rrbracket  _{q},$ where $q$ is
an odd prime power with $q\geq 5$ and $q\equiv 1$ $\left( mod\text{ }%
4\right) $ (\cite{ChenJ}).

\item $\llbracket  \frac{q^{2}+1}{10},\frac{q^{2}+1}{10}-2d+3,d;1\rrbracket  %
_{q},$ where $q$ is an odd prime power of the form $10m +3,$ $2 \leq d \leq
6m +2$ is even. (\cite{Lu2}).

\item $\llbracket  \frac{q^{2}+1}{10},\frac{q^{2}+1}{10}-2d+3,d;1\rrbracket  %
_{q},$ where $q$ is an odd prime power of the form $10m +7,$ $2 \leq d \leq
6m +4$ is even. (\cite{Lu2}).
\end{enumerate}
%\begin{acknowledgement}
%The author is very grateful to the reviewers for their comments and
%suggestions that improved the presentation and quality of this paper.
%\end{acknowledgement}
%\newpage


\begin{thebibliography}{99}
\bibitem{Ashikhmin} Ashikhmin, A., Litsyn, S., and Tsfasman, M. A.: Asymptotically good quantum codes. Phys. Rev. A \textbf{63,} 032311, (2001)
\bibitem{Aydin} Aydin, N., Siap, I. and Ray-Chaudhuri, D.K.: The structure of 1-generator quasi-twisted codes and new linear codes. Des. Codes Cryptogr. \textbf{24,} 313--326 (2001)
\bibitem{Brun} Brun, T., Devetak, I., Hsieh, M.H.: Correcting quantum errors with entanglement. Science \textbf{52,} 436 (2006)
\bibitem{Calderbank} Calderbank, A.R., and Shor, P.W.: Good quantum error-correcting codes exist. Phys. Rev. A \textbf{54,} 1098 (1996)
\bibitem{Calderbank1} Calderbank, A.R., Rains, E.M., Shor, P.W., and Sloane, N.J.A.: Quantum error correction via codes over $GF(4)$. IEEE Trans. Inform. Theory \textbf{44,} 1369-1387 (1998)
\bibitem{Chen} Chen, B., Ling, S. and Zhang, G.: Application of constacyclic codes to quantum MDS codes. IEEE Trans. Inf. Theory  \textbf{61,} 1474-1484  (2015)
\bibitem{ChenH} Chen, H.: Some good quantum error-correcting codes from algebraic-geometric codes. IEEE Trans. Inf. Theory \textbf{47,} 2059-2061 (2001)
\bibitem{ChenJ} Chen, J., Huang, Y., Feng, C. and Chen, R.: Entanglement-assisted quantum MDS codes constructed from negacyclic codes. Quantum Inform. Process. \textbf{16,} 303 (2017)
\bibitem{Fan} Fan, J., Chen, H., Xu, J.: Constructions of $q$-ary entanglement-assisted quantum MDS codes with minimum distance greater than $q+1.$ Quantum Inf. Comput. \textbf{16,} 0423–0434 (2016)
\bibitem{Fujiwara} Fujiwara, Y., Clark, D., Vandendriessche P., Boeck M.D., Tonchev V.D.: Entanglement-assisted quantum low-density parity-check codes. Phys. Rev. A \textbf{82,} 042338 (2010)
\bibitem{Grassl} Grassl, M.: Entanglement-assisted quantum communication beating the quantum singleton bound. AQIS, Taiwan (2016)
\bibitem{Guenda} Guenda, K., Jitman, S., Gulliver, T. A.: Constructions of good entanglement-assisted quantum error correcting codes. Des. Codes Cryptogr. \textbf{86,} 121-136 (2018)
\bibitem{Hsieh} Hsieh, M.H., Devetak, I., Brun, T.A.: General entanglement-assisted quantum error-correcting codes. Phys. Rev. A  \textbf{76,} 062313 (2007)
\bibitem{Hsieh1} Hsieh, M.H., Yen, W.T., Hsu, L.Y.: High performance entanglement-assisted quantum LDPC codes need little entanglement. IEEE Trans. Inf. Theory  \textbf{57,} 1761-1769 (2011)
\bibitem{Kai} Kai, X., and Zhu, S.: New quantum MDS codes from negacyclic codes. IEEE Trans. Inform. Theory  \textbf{59,} 1193-1197 (2013)
\bibitem{Kai1} Kai, X., Zhu, S., and Li, P.: Constacyclic codes and some new quantum MDS codes. IEEE Trans. Inform. Theory \textbf{60,} 2080-2086 (2014)
\bibitem{Ketkar} Ketkar, A., Klappenecker, A., Kumar, S., and Sarvepalli, P.K.: Nonbinary stabilizer codes over finite fields. IEEE Trans. Inform. Theory  \textbf{52,} 4892-4914 (2006)
\bibitem{Krishna} Krishna, A. and Sarwate, D.V.: Pseudocyclic maximum-distance-separable codes. IEEE Trans. Inform. Theory  \textbf{36,} 880-884 (1990)
\bibitem{Guardia} La Guardia, G.G.: On optimal constacyclic codes. Linear Algebra Appl. \textbf{496,} 594-610 (2016)
\bibitem{Lai} Lai, C.Y., Brun, T.A., Wilde, M.M.: Duality in entanglement-assisted quantum error correction. IEEE Trans. Inf. Theory \textbf{59,} 4020-4024 (2013)
\bibitem{Li} Li, R., Li, X., Guo, L.: On entanglement-assisted quantum codes achieving the entanglement-assisted Griesmer bound. Quantum Inform. Process. \textbf{14,} 4427-4447 (2015)
\bibitem{Lu} Lu, L., Li, R.: Entanglement-assisted quantum codes constructed from primitive quaternary BCH codes. Int. J. Quantum Inf. \textbf{12,} 1450015 (2014)
\bibitem{Lu1} Lu, L., Li, R., Guo, L., Ma, Y., and Liu, Y.: Entanglement-assisted quantum MDS codes from negacyclic codes. Quantum Inform. Process. \textbf{17,} 69 (2018)
\bibitem{Lu2} Lu, L., Ma, W., Li, R., Ma, Y., Liu, Y., and Cao, H.: Entanglement-assisted quantum MDS codes from constacyclic codes with large minimum distance. Finite Fields Th App, \textbf{53} 309-325 (2018)
\bibitem{MacWilliams} MacWilliams, F.J. and Sloane, N.J.A.: The theory of error-correcting codes. Elsevier.  (1977)
\bibitem{Qian} Qian, J., Zhang, L.: Nonbinary quantum codes derived from group character codes. Int. J. Quantum Inf. \textbf{10,} 1250042 (2012)
\bibitem{Qian1} Qian, J., Zhang, L.: New optimal subsystem codes. Discrete Math. \textbf{313,} 2451-2455 (2013)
\bibitem{Qian2} Qian, J., and Zhang, L.: On MDS linear complementary dual codes and entanglement-assisted quantum codes. Des. Codes Cryptogr. 1-8 (2017)
\bibitem{Shor} Shor, P.W.: Scheme for reducing decoherence in quantum memory. Phys. Rev.A \textbf{52,} 2493-2496 (1995)
\bibitem{Steane} Steane, A.M. Simple quantum error-correcting codes. Phys. Rev. A \textbf{54,} 4741 (1996)
\bibitem{Wilde} Wilde, M.M., Brun, T.A.: Optimal entanglement formulas for entanglement-assisted quantum coding. Phys. Rev. A \textbf{77,} 064302 (2008)
\bibitem{Xiaoyan} Xiaoyan, L.: Quantum cyclic and constacyclic codes. IEEE Trans. Inform. Theory \textbf{50,} 547-549 (2004)
\bibitem{Zhang} Zhang, T. and Ge, G.: Some new classes of quantum MDS codes from constacyclic codes. IEEE Trans. Inform. Theory \textbf{61,} 5224-5228 (2015)
\end{thebibliography}
\end{document}